%% file: main.tex
\documentclass[9.5pt,journal,compsoc]{IEEEtran}

\usepackage{algorithm}
\usepackage{algorithmic}
\usepackage{amssymb}
\usepackage{xspace}
\usepackage{comment}
\usepackage{cite}
\usepackage{subfigure}

\usepackage{helvet}
\usepackage[eulergreek]{sansmath}
\usepackage{pgfplots}
\usepackage{pgfplotstable}
\usepgfplotslibrary{
        groupplots,
    }
    
\pgfkeys{/pgf/number format/.cd,1000 sep={}}

\pgfplotsset{
width= 0.3\linewidth,
height= 4.3cm,
every axis/.append style={
label style={font=\sffamily\scriptsize},
tick label style = {font=\sansmath\scriptsize}
},
cycle list={
{only marks, mark options={color= blue, mark=o,scale=0.8}},
{only marks, mark options={color= red, mark=diamond ,scale=0.9}}
},
ylabel near ticks,
legend style={
anchor=north,
legend columns=1,
cells={anchor=west},
font=\sffamily\scriptsize
}
}

\usepackage{url}

\newcommand{\MYPROTOCOL}{CausalSpartanX\xspace}

\newtheorem{theorem}{\textbf{Theorem}}
\newtheorem{corollary}{\textbf{Corollary}}
\newtheorem{observation}{\textbf{Observation}}
\newtheorem{lemma}{\textbf{Lemma}}
\newtheorem{definition}{\textbf{Definition}}

\newtheorem{proof}{\textbf{Proof}}

\usepackage{graphicx}
\usepackage{color}

\newcommand{\deps}{$\ \textsf{dep} \ $}

\newcommand{\name}{CausalSpartanX\xspace}

\begin{document}
\title{\MYPROTOCOL: Causal Consistency and Non-Blocking Read-Only Transactions \thanks{A preliminary version of this work appeared in \cite{causalspartan}.}}


\author{
    \IEEEauthorblockN{
        Mohammad Roohitavaf\IEEEauthorrefmark{1},
        Murat Demirbas\IEEEauthorrefmark{2}, and
        Sandeep Kulkarni\IEEEauthorrefmark{1}}
    \IEEEauthorblockA{
        \begin{tabular}{cc}
            \begin{tabular}{@{}c@{}}
                \IEEEauthorrefmark{1}
                    Department of Computer Science and Engineering\\
Michigan State University\\
East Lansing, MI\\
Email: \{roohitav, sandeep\}@cse.msu.edu
            \end{tabular} & \begin{tabular}{@{}c@{}}
                \IEEEauthorrefmark{2}
                   Department of Computer Science and Engineering \\
University of Buffalo, SUNY \\
Buffalo, NY \\ 
Email: demirbas@buffalo.edu
            \end{tabular}
        \end{tabular}
    }
}

\IEEEtitleabstractindextext{%
\begin{abstract}
Causal consistency is an intermediate consistency model that can be achieved together with high availability and performance requirements even in presence of network partitions. In the context of partitioned data stores, it has been shown that implicit dependency tracking using timestamps is more efficient than explicit dependency tracking. Existing time-based solutions depend on monotonic psychical clocks that are closely synchronized. These requirements make current protocols vulnerable to clock anomalies. 
In this paper, we propose a new time-based algorithm, \MYPROTOCOL, that instead of physical clocks, utilizes Hybrid Logical Clocks (HLCs).  We show that using HLCs, without any overhead, we make the system robust on physical clock anomalies. 
This improvement is more significant in the context of query amplification, where a single query results in multiple GET/PUT operations. 
%
We also show that \MYPROTOCOL decreases the visibility latency for a given data item compared with existing time-based approaches. In turn, this reduces the completion time of collaborative applications where two clients accessing two different replicas edit same items of the data store. \MYPROTOCOL also provides causally consistent distributed read-only transactions. \MYPROTOCOL read-only transactions are non-blocking and require only one round of communication between the client and the servers. Also, the slowdowns of partitions that are unrelated to a transaction do not affect the performance of the transaction. Like previous protocols, \MYPROTOCOL assumes that a given client does not access more than one replica. We show that in presence of network partitions, this assumption (made in several other works) is essential if one were to provide causal consistency as well as immediate availability to local updates.

\end{abstract}
}

\maketitle

\input{introduction}

\input{background}
\input{motivation}

\input{hlc}

\input{protocol}

\input{results}

\input{impossibility}
\input{related}

\input{conclusion}

\bibliography{bib}

\newpage
\clearpage

\appendix

\input{correctness}

\end{document}

%% file: introduction.tex
\section{Introduction}

Geo-replicated data stores are one of the integral parts of today's Internet services. Service providers usually replicate their data on different nodes worldwide to achieve higher performance and durability. However, when we use this approach, the consistency among replicas becomes a concern. In an ideal situation, any update to any data item instantaneously becomes visible in all replicas. This model of consistency is called \textit{strong consistency}. 
Unfortunately, it is impossible to achieve strong consistency without sacrificing the availability when we have network partitions; 
the CAP theorem \cite{cap} implies that in presence of network partitions, we cannot have strong consistency and availability together. Even in the absence of network partitions, strong consistency comes with its performance overhead \cite{pacelc}.

Due to the availability and performance overhead of the strong consistency, many systems use \textit{eventual consistency} \cite{ec}. In this consistency model, as the name suggests, the only guarantee is that replicas become consistent "eventually". We can implement always-available services under this consistency model. However, it may expose clients to anomalies, and application programmers need to consider such anomalies. To understand how eventual consistency may lead to an anomaly, consider the following example from \cite{cops}: Suppose in a social network, Alice uploads a photo and then adds it to an album. Under the eventual consistency model, a remote replica may update the album before writing the photo. That scenario is not desirable, as the album is pointing to a photo which is not visible to clients. Despite such anomalies, because of the availability and performance benefits, some distributed data stores (e.g., Dynamo \cite{dynamo}) use eventual consistency.

\textit{Causal consistency} is an intermediate consistency model.
Causal consistency requires that the effect of an event can be visible only when the effect of its causal dependencies is visible. The causal dependency captures the notion of happens-before relation defined in \cite{lamport}. Under this relation, any event by a client depends on all previous events by that client. Thus, in our example, adding the reference to the album depends on the event of uploading the photo. Thus, no replica can update the album before writing the photo. Causal consistency 
is achievable with availability even in the presence of networks partitions. 


To \textit{guarantee} causal consistency for an \textit{always-available} system, we need some mechanism to track and check causal dependencies before making a version visible in a remote replica. Tracking and checking dependencies are especially challenging for partitioned systems where each replica consists of several machines. 
Existing work on causal consistency for replicated and partitioned data stores can be classified into those where dependencies are tracked and checked explicitly (e.g., COPS \cite{cops}), and those that do it implicitly (e.g., GentleRain \cite{gentleRain}). The former suffers from overhead resulted from high metadata and message-complexity that is avoided by the latter. For example, GentleRain \cite{gentleRain} uses physical clock timestamps to track causal dependencies that results in much lower metadata overhead and message complexity compared with COPS \cite{cops}.


Although GentleRain reduces the overhead of tracking and checking dependencies, the physical clocks used for timestamping the versions must be monotonic and synchronized with each other. Specifically, it requires that clocks are strictly increasing. This may be hard to guarantee if the underlying service such as NTP causes non-monotonic updates to POSIX time \cite{NTPbad} or suffers from leap seconds \cite{leapsecond,leapsecond2}. In addition, as we will see, the clock skew between physical clocks of partitions may lead to cases where GentleRain must intentionally delay write operations.

The issue of clock anomalies is intensified in the context of query amplification, where a single query results in many (possibly 100s to 1000s) GET/PUT operations \cite{facebook}. In this case, the delays involved in each of these operations contribute to the total delay of the operation, and can substantially increase the response time for the clients.

Our goal in this paper is to analyze the effect of clock anomalies to develop a causally consistent data store that is resistant to clock skew among servers. This will allow us to ensure that high performance is provided even if there is a clock skew among servers. It would obviate the need for all servers in a data center to be co-located for the sake of reducing clock anomalies.

To achieve this goal, we develop \MYPROTOCOL that is based on the structure of GentleRain but utilizes Hybrid Logical Clocks (HLCs)\cite{hlc}. HLCs combine the logical clocks \cite{lamport} and physical clocks. In particular, these clocks assign a timestamp $hlc$ for event $e$ such that if $e$ happened before $f$ (as defined in \cite{lamport}), then $hlc.e < hlc.f$. Furthermore, the value of $hlc$ is very close to the physical clock and is backward compatible with NTP clocks \cite{ntp}. 

In addition to providing causal consistency for basic PUT and GET operations, \name also provides causally consistent read-only transactions as another powerful abstraction that can significantly help application developers. \name read-only transaction reads a set of keys from a causally consistent snapshot of the system that is also causally consistent with the client previous reads. \name read-only transaction algorithm is non-blocking, i.e., the servers involved in the transaction do not need to wait for an external event before reading the values of the keys requested by the transaction. Also, unlike previous protocols such as \cite{cops}, \cite{gentleRain}, and \cite{cont}, \name requires only one-round of communication between the client and the servers. An important advantage of \name read-only transactions (over existing works such as \cite{gentleRain}) is that the performance of \name read-only transaction is not affected by a slow unrelated partition (i.e., a partition that does not host any key asked by the transaction).  


Similar to \cite{gentleRain, cops, eiger, orbe}, we assume that clients only accesses one replica during their execution (i.e. clients are \textit{sticky}). Since this assumption is standard in the literature, we investigate its necessity. We observe that this assumption is essential if we want to provide causal consistency while ensuring that all local updates are immediately visible. In other words, we show that if the clients are not sticky and a replica makes local updates visible immediately then it is impossible to provide causal consistency with availability in presence of network partitions. This impossibility result is different than the existing impossibility result that requires sticky clients for read-your-writes \cite{peter} that is a part of causal consistency.
Specifically, the proof of impossibility result in \cite{peter} relies on the inability of clients to cache their updates locally. 
By contrast, our impossibility result holds even if clients can cache their updates.

{\bf Contributions of the paper. } \
\begin{itemize}
  \setlength\itemsep{0mm}

\item We present our \name protocol that provides causally consistent basic PUT and GET operations as well as non-blocking causally consistent read-only transactions.

\item We show that in the presence of clock anomalies, \MYPROTOCOL reduces the latency of PUT operations compared with that of GentleRain. Moreover, the performance or correctness of \MYPROTOCOL is unaffected by clock anomalies. 


\item We demonstrate that  \MYPROTOCOL is especially effective to deal with delays associated with query amplification. 

\item We demonstrate that \MYPROTOCOL reduces the update visibility latency. This is especially important in collaborative applications associated with a data store. For example, in an application where two clients update a common variable (for example, the bid price for an auction) based on the update of the other client, \MYPROTOCOL reduces the execution time substantially. 
\item We show that \name non-blocking read-only transaction algorithm provides better performance compared with GentleRain by not being affected by slowdowns of unrelated partitions. 

\item We show that using HLC instead of physical clocks does not have any overhead.
\item We demonstrate the efficiency provided by our approach by performing experiments on cloud services provided by Amazon Web Services\cite{aws}.
\item We provide an impossibility result that shows stickiness of the clients is necessary for a causally consistent data store that immediately makes local updates visible. We note that unlike the existing result \cite{snow}, our impossibility theorem still holds even when clients can cache their updates.
\end{itemize}

{\bf Organization of the paper. } \
In Section \ref{sec:background}, we define our system architecture and the notion of causal consistency. In Section \ref{sec:motivation}, we discuss, in detail, the issues of clock anomalies that we want to address. In Section \ref{sec:hlc}, we provide a brief overview of HLCs from \cite{hlc}. Section \ref{sec:protocol} provides our \MYPROTOCOL for basic operation and Section \ref{sec:rotx} provides \name protocol for read-only transactions. Section \ref{sec:results} provides our experimental results. We provide the impossibility result in Section \ref{sec:impossibility}. In Section \ref{sec:related} we discuss related work.
Finally, Section \ref{sec:conclusion} concludes the paper.

%% file: background.tex
\section{Background}
\label{sec:background}
In this section, we focus on the system architecture, assumptions, and the data model considered in this paper. We also provide an intuitive description of causal consistency.

\subsection{Architecture and Data Model}
\label{sec:arch}    
We focus on the system architecture and assumptions that are the same as those assumed in \cite{cops,eiger, orbe, gentleRain}. We consider a data store whose data is fully replicated into $M$ data centers (i.e., replicas) where each data center is partitioned into $N$ partitions (see Figure \ref{fig:arch}). 
Like \cite{cops,eiger, orbe, gentleRain}, we assume that a client does not access more than one data center. We prove the necessity of this assumption in Section \ref{sec:impossibility}. There might be network failures \textit{between} data centers that cause network partitions. We assume network failures do not happen \textit{inside} data centers. Thus, partitions inside a data center can always communicate with each other.

\begin{figure}
\begin{center}
\includegraphics[scale=0.7]{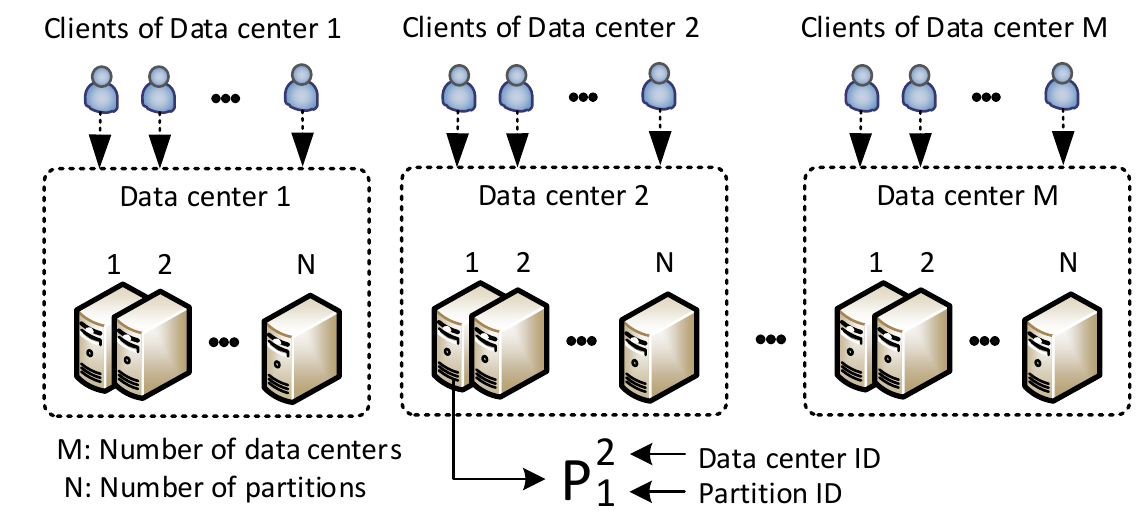}
\caption{A system consisting of $M$ data centers (replicas) each of which consists of $N$ partitions. $p^m_n$ denotes $n$th partition in $m$th data center.}
\label{fig:arch}
\end{center}
\end{figure}

We assume multi-version key-value stores that store several versions for each key. We consider three operations for a key-value store: $PUT(k, val)$, $GET(k)$, and $ROTX(K)$ where  $PUT(k, val)$ writes new version with value $val$ for item with key $k$, $GET(k)$ reads the value of an item with key $k$, and $ROTX(K)$ is a read-only transaction that reads a set of keys $K$.

\subsection{Causal Consistency}
\label{sec:causalConsistency}
Causal consistency is defined based on the happens-before relation between events \cite{lamport}.  In the context of key-value stores, we define happens-before relation as follows: 

\begin{definition} [Happens-before]
\label{def:happens}
Let $a$ and $b$ be two events. We say $a$ happens before $b$, and denote it as $a \rightarrow b$ iff: 
\begin{itemize}
\setlength{\itemsep}{0mm}
\item $a$ and $b$ are two events by the same client, and $a$ happens earlier than $b$, or
\item $b$ reads a value written by $a$, or
\item there is another event $c$ such that $a \rightarrow c$ and $c \rightarrow b$. 
\end{itemize}
\end{definition}

Now, we define causal dependency as follows: 

\begin{definition} [Causal Dependency]
\label{def:causal_dependecy}
Let $v_1$ be a version of key $k_1$, and $v_2$ be a version of key $k_2$. We say $v_1$ causally depends on $v_2$, and denote it as $v_1 \deps v_2$ iff (event of writing $v_2$) $\rightarrow$ (event of writing $v_1$). 
\end{definition}

For example, suppose that a user writes a comment for a post on a social network. Since the user, first reads the post, and then writes the comment, the comment causally depends on the post.


A data store is causally consistent if it satisfies these conditions: 1) when a client reads a version, it always remains visible to the client, 2) writes by a client must be immediately visible to the client, and 3) when a version is visible,  all of its causal dependencies are also visible.


%% file: motivation.tex
\section{Causal Consistency and Physical Clock Anomalies}
\label{sec:motivation}

To guarantee causal consistency, we need to track the causal relation between the versions. Tracking dependencies implicitly using timestamps is much more efficient than tracking dependencies explicitly (i.e., via a list of dependencies) \cite{gentleRain}. In this section, we identify the issues caused by clock anomalies in providing causal consistency using the physical timestamps in protocols such as GentleRain \cite{gentleRain}. First, in Section \ref{sec:gentleRain}, we identify how causal consistency is achieved in \cite{gentleRain} with synchronized physical clocks. Next, in Section \ref{sec:sensitivity}, we identify why delays have to be introduced in PUT operations to satisfy causal consistency. In Section \ref{sec:query}, we identify why the effect of latency introduced in PUT operations causes a significant problem to queries that result in multiple GET/PUT operations for a given query. 

\subsection{Using Physical Clocks to Achieve Causal Consistency}
\label{sec:gentleRain}

GentleRain assigns each version a timestamp equal to the value of the physical clock of the partition where the write of the version occurs. We denote the timestamp assigned to version $X$ by $X.t$. GentleRain assigns timestamps such that following condition is satisfied:

\begin{itemize}
\item [$C1:$] If version $X$ of object $x$ depends on version $Y$ of object $y$, then $Y.t < X.t$.
\end{itemize} 



Also, each partition in the data center periodically computes a variable called  Global Stable Time (GST) (through communication with other partitions) such that the following condition is satisfied: 
\begin{itemize}
\item [$C2:$] When GST in a node has a certain value $T$, then all versions with timestamps smaller than or equal to $T$ are visible in the data center.
\end{itemize}

When a client performs $GET(k)$, the partition storing $k$, returns the newest version $v$ of $k$ which is either created locally or has a timestamp no greater than GST. According to conditions $C1$ and $C2$ defined above, any version which is a dependency of $v$ is visible in the local data center and causal consistency is satisfied.

\subsection{Sensitivity on physical clock and clock synchronization} 
\label{sec:sensitivity}

To satisfy the condition $C1$, in some cases, it may be necessary to wait before creating a new version. Specifically, if a client has read/written a key with timestamp $t$, then any future PUT operation the client invokes must have a timestamp higher than $t$. Hence,  the client sends the timestamp $t$ of the last version that it has read/written together with a PUT operation. The partition receiving this request first waits until its physical clock is higher than $t$ before creating the new version. This wait time, as we observed in our experiments, is proportional to the clock skew between servers. In other words, as the physical clocks of servers drift from each other, the incidence and the amount of this wait period increases. 

In addition, in the approach explained in Section \ref{sec:gentleRain}, the physical clocks cannot go backward. To illustrate this, consider a system consisting of two data centers $A$ and $B$. Suppose GSTs in both data centers are 6. That means, both data centers assume all versions with timestamps smaller than 6 are visible (condition $C2$). Now, suppose the physical clock of one of the servers in data center $A$ goes backward to 5. In this situation, if a client writes a new version at that server, condition $C2$ is violated, as the version with timestamp 5 has not arrived in data center $B$, but its GST is 6 which is higher than 5. 


\subsection{Query Amplification}
\label{sec:query}

The sensitivity issue identified in Section \ref{sec:sensitivity} is made worse in practice because a single end user request usually translates to many possibly causally dependent internal queries. This phenomenon is known as \textit{query amplification} \cite{facebook}. In a system like Facebook, query amplification may result in up to thousands of internal queries for a single end user request \cite{facebook}. Typically, an end user submits the request to a web server. The web server performs necessary internal queries and then responds to the end user. This implies that the web server needs to wait for all internal queries before responding to the end user request. Thus, any delay in any of internal queries will cumulatively affect the end user experience \cite{facebook}.


\MYPROTOCOL solves the issues identified in this section by ensuring that no delays are added to PUT operations. Therefore, \MYPROTOCOL is unaffected by clock skew even in the presence of query amplification. 

%% file: hlc.tex
\section{Hybrid Logical Clocks}
\label{sec:hlc}

In this section, we recall HLCs from \cite{hlc}. HLC combines logical and physical clocks to leverage key benefits of both. The HLC timestamp of event $e$, denoted as $hlc.e$, is a tuple $\langle l.e, c.e\rangle$. The first component, $l.e$, is the value of the physical clock, and represents our best approximation of the global time when $e$ occurs. The second component, $c.e$, is a \textit{bounded} counter that is used to capture causality whenever $l.e$ is not enough to capture causality. Specifically, if we have two events $e$ and $f$ such that $e$ happens-before $f$ (see Definition \ref{def:happens}), and $l.e = l.f$, to capture causality between $e$ and $f$, we set $c.e$ to a value higher than $c.f$. Although we increase $c$, as it is proved in \cite{hlc}, the theoretical maximum value of $c$ is $O(n)$ where $n$ is the number of processes. In practice, this value remains very small. In addition to HLC timestamps, each process $a$ maintains an HLC clock $\langle l.a, c.a\rangle$. For completeness, we recall algorithm of HLC from \cite{hlc} below. 

\begin{algorithm}[t] 
{
\footnotesize
\caption{HLC algorithm from \cite{hlc} }
\label{alg:hlc}
\begin{algorithmic} [1]
\STATE \textbf{Upon sending a message or local event by process $a$} 
\STATE \hspace{3mm} $l'.a = l.a$
\STATE \hspace{3mm} $l.a = max(l'.a, pt.a)$ //tracking maximum time event, $pt.a$ is physical time at a
\STATE \hspace{3mm} \textbf{if} $(l.a = l'.a) \ c.a = c.a + 1$ //tracking causality
\STATE \hspace{6mm} \textbf{else} $c.a = 0$
\STATE \hspace{3mm} Timestamp event with $l.a,c.a$

\STATE \vspace{5mm} \textbf{Upon receiving message $m$ by process $a$} 
\STATE \hspace{3mm}  $l'.a = l.a$
\STATE \hspace{3mm}  $l.a = max(l'.a, l.m, pt.a)$ //$l.m$ is $l$ value in the timestamp of the message received
\STATE \hspace{3mm}  \textbf{if} $(l.a = l'.a = l.m)$ then $c.a = max(c.a, c.m) + 1$
\STATE \hspace{3mm}  \textbf{else if} $(l.a = l'.a)$ then $c.a = c.a + 1$
\STATE \hspace{3mm}  \textbf{else if} $(l.a = l.m)$ then $c.a := c.m + 1$
\STATE \hspace{3mm}  \textbf{else} $c.a = 0$
\STATE \hspace{3mm}  Timestamp event with $l.a,c.a$
\end{algorithmic}
}
\end{algorithm}

HLC satisfies logical clock property that allows us to capture (one-way) causality between two events. Specifically, if $e$ happens-before $f$, then $hlc.e < hlc.f$\footnote{$hlc.e < hlc.e$ iff $l.e < l.f \vee (l.e = l.f \wedge c.e < c.f)$.}.  This implies that if $hlc.e = hlc.f$, then $e$ and $f$  are (causally) concurrent. At the same time, just like physical clock, HLC increases spontaneously, and it is close to the physical clock. Thus, it can be used to take snapshot at a given physical time. Moreover, since the timestamps are close to the physical clocks, we can use timestamps to achieve last-writ-wins conflict resolution in case of concurrent updates \cite{session}.

%% file: protocol.tex
\section{\MYPROTOCOL Basic Protocol}
\label{sec:protocol}


One way to get around the issue of PUT latency identified in Section \ref{sec:motivation} is as follows: Suppose that a client has read a value written at time $t$ and it wants to perform a new PUT operation on a server whose time is less than $t$. To satisfy $C1$, in \cite{gentleRain}, PUT operation is delayed until the clock of the server was increased beyond $t$. Another option is to change the clock of the server to be $t+1$. 
However, changing the physical clock is undesirable; it would lead to a violation of clock synchronization achieved by protocols such as NTP. It would also have unintended consequences for applications using that clock.

Using HLC in this problem solves several problems associated with changing the physical clock. Specifically, HLC is a logical clock and can be changed if needed. In the scenario described in the previous paragraph, this would be achieved by increasing the $c$ value which is still guaranteed to stay bounded \cite{hlc}.
At the same time, HLC is guaranteed to be close to the physical clock. Hence, it can continue to be used in place of the physical clock. Also, HLC uses the physical clock as a read-only variable thereby ensuring that it does not affect protocols such as NTP. For these reasons, \MYPROTOCOL uses HLC. 

Another important improvement in \MYPROTOCOL is the use of Data center Stable Vectors (DSVs) instead of GSTs. DSVs are vectors that have an entry for each data center. If $DSV[j]$ equals $t$ in data center $i$, then it implies that all writes performed at data center $j$ before time $t$ have been received by data center $i$. 
DSVs reduce update visibility latency and allow collaborative clients to work quickly in the presence of some \textit{slow} replicas.
Next, we focus on different parts of the \MYPROTOCOL protocol.

\subsection{Client-Side}

A client $c$ maintains a set of pairs of data center IDs and HLC timestamps called dependency set, denoted as $DS_c$\footnotemark. For each data center $i$, there is at most one entry $\langle i, h \rangle$ in $DS_c$ where $h$ specifies the maximum timestamp of versions read by client $c$ originally written in data center $i$. For a given PUT request, this information is provided by the client so that the server can guarantee causal consistency. 
A client $c$ also maintains $DSV_c$ that is the most recent $DSV$ that the client is aware of. 

\footnotetext{This could be maintained as part of client library as in \cite{cops} so that the effective interface of the client does not have to explicitly maintain this information. Alternatively, this could also be maintained by the server for each client. For sake of simplicity, in our discussion, we assume that this information is provided by the client.}

Algorithm \ref{alg:client} shows the algorithm for the client operations. For a GET operation, the client sends the key that it wants to read together with its $DSV_c$ by sending  $\langle \textsc{GetReq} \ k, DSV_c\rangle$ message to the server where key $k$ resides. In the response, the server sends the value of the requested key together with a list of dependencies of the returned value, $ds$, and the DSV in the server, $dsv$. The client, then, first updates its DSV, and next updates its $DS_c$ by calling $maxDS$ as follows: for each $\langle i, h\rangle \in ds$ if currently there is an entry $\langle i, h'\rangle$ in $DS_c$, it replaces $h'$ with the maximum of $h$ and $h'$. Otherwise, it adds $\langle i, h\rangle$ to the $DS_c$. 

For a PUT operation, the client sends the key that it wants to write together with the desired value and its $DS_c$. In response, the server sends the timestamp assigned to this update together with the ID of the data center. The client then updates its $DS_c$ by calling $maxDS$.

\begin{algorithm} 
{
\footnotesize
\caption{Client operations at client $c$}
\label{alg:client}
\begin{algorithmic} [1]

\STATE \textbf{GET} (key $k$) 

\STATE \hspace{3mm}  send $\langle \textsc{GetReq} \ k, DSV_c\rangle$  to server
\STATE \hspace{3mm}  receive $\langle \textsc{GetReply} \ v, ds, dsv\rangle$
\STATE \hspace{3mm} \label{line:client_update_dsv} $DSV_c \leftarrow max (DSV_c, dsv)$
\STATE \hspace{3mm} \label{line:client_update_ds} $DS_c \leftarrow  maxDS(DS_c, ds)$

\RETURN $v$

\STATE \vspace{5mm} \textbf{PUT} (key $k$, value $v$)
\STATE \hspace{3mm} send $\langle \textsc{PutReq} \ k,v,DS_c \rangle$ to server
\STATE \hspace{3mm} receive $\langle \textsc{PutReply} \ ut, sr \rangle$ 
\STATE \hspace{3mm} \label{line:put_update_ds}$DS_c \leftarrow maxDS(DS_c, \{\langle sr, ut\rangle\})$

\STATE \vspace{5mm} \textbf{maxDS} (dependency set $ds_1$, dependency set $ds_2$)
\STATE \hspace{3mm} \textbf{for} each $\langle i, h\rangle \in ds_2$ 
\STATE\hspace{6mm} \textbf{if} {$\exists \langle i, h'\rangle \in ds_1$} 
\STATE \hspace{9mm} $ds_1 \leftarrow ds_1 - {\langle i, h' \rangle}$
\STATE \hspace{9mm}  $ds_1 \leftarrow \langle i, max (h, h')\rangle$
\STATE \hspace{6mm} \textbf{else}
\STATE \hspace{9mm} $ds_1 \leftarrow ds_1 \cup \{\langle i, h \rangle\} $
\RETURN $ds_1$
\end{algorithmic}
}
\end{algorithm}

\subsection{Server-Side}

In this section, we focus on the server-side of the protocol. We have $M$ data centers (i.e., replicas) each of which with $N$ partitions (i.e., servers). We denote the $n$th partition in $m$th replica by $p^m_n$ (see Figure \ref{fig:arch}). 
We denote the physical clock at partition $p^m_n$ by $PC^m_n$. Each partition $p^m_n$ stores a vector of size $M$ (one entry for each data center) of (HLC) timestamps denoted by  $VV^m_n$. For $k \neq m$, $VV^m_n[k]$ is the latest timestamp received from server $p^k_n$ by server $p^m_n$. $VV^m_n[m]$ is the highest timestamp assigned to a version written in partition $p^m_n$. Partitions inside a data center, periodically share their $VVs$ with each other, and compute DSV as the entry-wise minimum of $VVs$. $DSV^m_n$ is the DSV computed in server $p^m_n$.

For each version, in addition to the key and value, we store some additional metadata including the (HLC) time of creation of the version, $ut$, and the source replica, $sr$, where the version has been written, and a set of dependencies, $ds$, similar to dependency sets of clients. Note that  $ds$ has at most one entry for each data center. 

Algorithm \ref{alg:server1} shows the algorithm for PUT and GET operations at the server-side. Upon receiving a GET request ($\textsc{GetReq}$), the server first updates its DSV if necessary using DSV value received from the client (see Line \ref{line:get:takeMaxDSV} of Algorithm \ref{alg:server1}). 
 After updating DSV, the server finds the latest version of the requested key that is either written in the local data center, or all of its dependencies are visible in the data center.  To check this, the server compares the DS of the key with its DSV. Note that to find the \textit{latest} version, the server uses the last-writer-wins conflict resolution function that breaks ties by data center IDs as explained in Section \ref{sec:causalConsistency}. 
After finding the proper value, the server returns the value together with the list of dependencies of the value, and its DSV in a $\textsc{GetReply}$ message. The server also includes the version being returned in the dependency list in Line \ref{line:includeUt} of Algorithm \ref{alg:server1} by calling the same $maxDS$ function as defined in Algorithm \ref{alg:client}. 

A major improvement in \MYPROTOCOL over GentleRain is providing wait-free PUT operations. Once server  $p^m_n$ receives a PUT request, the server first updates its DSV with the DS value received from the client. The server, next, updates $VV^m_n[m]$  to calling $updateHLC$ function. It uses maximum number between $ds$ values and $DSV^m_n[m]$ as $updateHLC$ argument. This will guarantee that the new $VV^m_n[m]$ to be higher than this maximum and capture causality. Next, the server creates a new version for the key specified by the client and uses the current $VV^m_n[m]$ value for its timestamp. The server sends back the assigned timestamp $d.ut$ and data center ID $m$ to the client in a $\textsc{PutReply}$ message. 

 

Upon creating a new version for an item in one data center, we send the new version to other data centers via replicate messages. Upon receiving a $\langle Replicate \ d\rangle$ message from server $p^k_n$, the receiving server $p^m_n$ adds the new version to the version chain of the item with key $d.k$. The server also updates the entry for server $p^k_n$ in its version vector. Thus, it sets $VV^m_n[k]$ to $d.ut$. 

Algorithm \ref{alg:server2} shows the algorithm for updating DSVs. As mentioned before, partitions inside a data center periodically update their DSV values. Specifically, 
every $\theta$ time, partitions share their VVs with each other and compute DSV as the entry-wise minimum of all VVs (see Line \ref{line:computeDSV} of Algorithm \ref{alg:server2}). Broadcasting VVs has a high overhead. Instead, we efficiently compute DSV over a tree like the way GST is computed in \cite{gentleRain}. Specifically, each node upon receiving VVs of its children computes entry-wise minimum of the VVs and forwards the result to its parent. The root server computes the final DSV and pushes it back through the tree. Each node, then, updates its DSV upon receiving DSV from its parent. Algorithm \ref{alg:server2} also shows the algorithm for the heartbeat mechanism. Heartbeat messages are sent by a server if the server has not sent any replicate message for a certain time $\Delta$. The goal of heartbeat messages is updating the knowledge of the peers of a partition in other data centers (i.e., updating $VV$s). 

\begin{algorithm} 
{
\footnotesize
\caption{PUT and GET operations at server $p^m_n$}
\label{alg:server1}
\begin{algorithmic} [1]
\STATE \textbf{Upon} receive $\langle \textsc{GetReq} \ k, dsv\rangle$
\STATE \hspace{3mm} \label{line:get:takeMaxDSV}$DSV^m_n \leftarrow max(DSV^m_n, dsv)$
\STATE \hspace{3mm}  \label{line:get:obtain} obtain latest version $d$ (the version with lexicographically highest value $\langle ut, sr\rangle$) from version chain of key $k$ s.t. 
\begin{itemize}
\item $d.sr = m$, or
\item for any $\langle i, h\rangle \in d.ds$, $h \leq DSV^m_n[i]$
\end{itemize}
\STATE \hspace{3mm} \label{line:includeUt} $ds \leftarrow maxDS(d.ds, \{\langle d.sr, d.ut\rangle\})$
\STATE \hspace{3mm} send $\langle\textsc{GetReply} \  d.v, ds, DSV^m_n\rangle$ to client

\STATE \vspace{5mm} \textbf{Upon}  receive $\langle \textsc{PutReq} \ k, v, ds \rangle$
\STATE \hspace{3mm} \label{line:put_max_dsv} $DSV^m_n \leftarrow maxDS (DSV^m_n, ds)$
\STATE \hspace{3mm} \label{line:takeMaxforDt} $dt \leftarrow$ max value in $\{ds.values \cup \{DSV^m_n[m]\}\}$  
\STATE \hspace{3mm}  $updateHCL(dt)$
\STATE \hspace{3mm}  Create new item $d$
\STATE \hspace{3mm}  $d.k \leftarrow k$ 
\STATE \hspace{3mm}  $d.v \leftarrow v$
\STATE \hspace{3mm}  $d.ut \leftarrow VV^m_n[m]$
\STATE \hspace{3mm}  $d.sr \leftarrow m$ 
\STATE \hspace{3mm}  $d.ds \leftarrow ds$ 
\STATE \hspace{3mm}  insert $d$ to version chain of $k$
\STATE \hspace{3mm}  send $\langle \textsc{PutReply} \ d.ut , m\rangle$ to client
\STATE \hspace{3mm}  \textbf{for} each server $p^k_n, k \in \{0 \ldots M-1\}, k \neq m$ \textbf{do}
\STATE \hspace{6mm}  \label{line:sendReplicate} send $\langle \textsc{Replicate} \ d \rangle$ to $p^k_n$

\STATE \vspace{5mm} \textbf{Upon}  receive $\langle \textsc{Replicate} \  d\rangle$ from $p^k_n$
\STATE \hspace{3mm}  insert $d$ to version chain of key $d.k$
\STATE \hspace{3mm} $VV^m_n  [k] \leftarrow d.ut$

\STATE \vspace{5mm} \textbf{updateHLCforPut ($dt$)}
\STATE \hspace{3mm} $l' \leftarrow VV^m_n[m].l$
\STATE \hspace{3mm} $VV^m_n[m].l \leftarrow max (l', PC^m_n, dt.l)$
\STATE \hspace{3mm} \textbf{if} $(VV^m_n[m].l = l' = dt.l) \ \ VV^m_n[m].c \leftarrow max(VV^m_n[m].c ,dt.c)+1$
\STATE \hspace{3mm} \textbf{else if} $(VV^m_n[m].l = l') \  \ VV^m_n[m].c \leftarrow VV^m_n[m].c + 1$
\STATE \hspace{3mm} \textbf{else if} $(VV^m_n[m].l = dt.l) \  \ VV^m_n[m].c \leftarrow dt.c + 1$
\STATE \hspace{3mm} \textbf{else} $VV^m_n[m].c \leftarrow 0$

\end{algorithmic}
}
\end{algorithm}

\begin{algorithm} 
{
\footnotesize
\caption{HEARTBEAT and DSV computation operations at server $p^m_n$}
\label{alg:server2}
\begin{algorithmic} [1]
\STATE \textbf{Upon}  every $\theta$ time
\STATE \hspace{3mm} \label{line:computeDSV} $DSV^m_n \leftarrow max(DSV^m_n,\ $entry-wise $min_{j=1}^{N} (VV^m_j))$

\STATE \vspace{5mm} \textbf{Upon} every $\Delta$ time
\STATE \hspace{3mm} \textbf{if} there has not been any replicate message in the past $\Delta$ time 
\STATE \hspace{6mm}  $updateHCL()$
 \STATE \hspace{6mm}  \textbf{for} each server $p^k_n, k \in \{0 \ldots M-1\}, k \neq m$ \textbf{do}
\STATE \hspace{9mm}  send $\langle \textsc{Heartbeat} \ HLC^m_n \rangle$ to $p^k_n$

\STATE \vspace{5mm} \textbf{Upon}  receive $\langle \textsc{Heartbeat} \ hlc \rangle$ from $p^k_n$
\STATE \hspace{3mm}  $VV^m_n  [k] \leftarrow hlc$

\STATE \vspace{5mm} \textbf{updateHLC} ()
\STATE \hspace{3mm} $l' \leftarrow VV^m_n[m].l$
\STATE \hspace{3mm} $VV^m_n[m].l \leftarrow max(VV^m_n[m].l, PC^m_n)$
\STATE \hspace{3mm} \textbf{if} $(VV^m_n[m].l = l')\ $ $VV^m_n[m].c \leftarrow VV^m_n[m].c + 1$
\STATE \hspace{3mm} \textbf{else} $VV^m_n[m].c \leftarrow 0$

\end{algorithmic}
}
\end{algorithm}

\input{transaction}

%% file: transaction.tex
\section{\name ROTX Protocol}

In this section, we want to focus on causally consistent read-only transactions. We, first, discuss the motivation for this operation, and, then, present an algorithm to provide this operation in \name. 

\label{sec:rotx}

\subsection{Motivation}
A causally consistent read-only transaction is a powerful abstraction that can significantly help application developers when working with replicated data stores. A read-only transaction allows application developers to read a set of keys such that the returned versions of the key values are causally consistent with each other as well as with previous reads of the application. To understand the benefit of such abstraction, consider the following example. 

Consider a social network such as Facebook where profile pictures are always public. Alice wants to update her profile picture, but she does not want Bob to see her new picture. Since profile pictures are public, the only way for Alice to hide her picture from Bob is to completely block Bob. Thus, she first blocks Bob and then updates her picture (we call this change scenario 1). A data store with only causally consistent PUT and GET operations (without read-only transactions) guarantees that no matter which replicas Alice and Bobs are connected to, the new Alice's picture is visible only when Bob is blocked by Alice. However, it is \textit{not} enough to protect Alice's privacy. Suppose the application first reads Bob's status and finds it unblocked, then it reads Alice's new profile picture. Since it found Bob unblocked, it shows Alice's new picture to Bob that is not acceptable. With a causally consistent read-only transaction, the application could read both values such that they are causally consistent with each other. Thus, the application either 1) reads Alice's old picture, or 2) it reads Alice's new picture, but finds Bob blocked, and both of these cases are acceptable. 

Now, suppose after some time, Alice changes her mind. She changes her photo to the old one and then unblocks Bob (we call this change scenario 2). A causally consistent read-only transaction still protects Alice's privacy after this change. With two independent causally consistent GET operations it is impossible to protect Alice's privacy in both cases. Specifically, based on the order that the application issues the two GET operations, it violates Alice's privacy either in scenario 1 or scenario 2. 

\begin{figure}
\begin{center}
\includegraphics[scale=0.092]{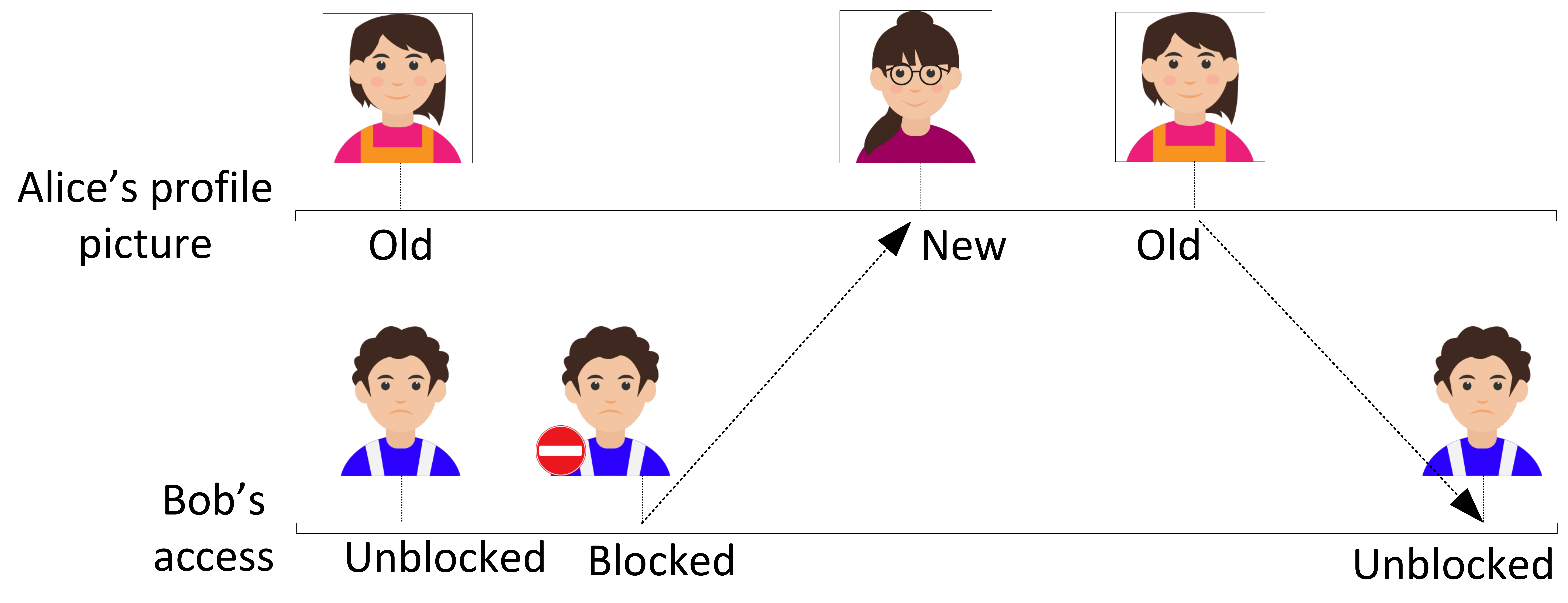}
\caption{Updating profile picture in a social network while blocking another user.}
\label{fig:rotx}
\end{center}
\end{figure}

\subsection{ROTX Algorithm}


In this section, we provide an algorithm for ROTX, that never blocks, and requires only one round of communication between the client and the servers. Algorithm \ref{alg:rotx} shows both sides of our algorithm for ROTX. Upon request of an ROTX operation from the application, the client sends \textsc{ROTX} request with the set of requested keys by the application together with its $DSV_c$ and $DS_c$ to one of the servers hosting one of the requested keys to read. In the response, the server sends the values of the requested keys, updated DSV, and the set of dependencies of the returned values. Upon receiving the response, the client updates its $DSV_c$ and $DS_c$ with the new values returned by the server, and return the received values from the server to the application.

On the server-side, upon receiving an \textsc{ROTX} request from a client, first updates it DSV by $dsv$ and $ds$ values received from the client. The server, next, uses the current value of its DSV as the snapshot vector, $sv$. Then, for each requested key $k$, the server sends a \textsc{SliceReq} request with $sv$ to the partition hosting $k$. Upon receiving response to its  \textsc{SliceReq} message, the server reads the returned value, and updates $ds$. Once the server learned the value of all requested keys, it sends the response to the client together with its new DSV and the update $ds$. 

Upon receiving a \textsc{SliceReq} request, a partition first updates the local entry of its DSV by the corresponding entry in $sv$. Next, it retrieves the most recent version $s$ of the request key such that all dependencies of the version are smaller than their corresponding entries in $sv$, and the update time of the version, $d.ut$, is less than or equal to the $sv[d.sr]$. Then, the partition returns the response to the requesting server with the value of the key and the set of dependencies.

\begin{algorithm} 
{
\footnotesize
\caption{Algorithm for ROTX}
\label{alg:rotx}
\begin{algorithmic} [1]
\STATE //at client $c$
\STATE \textbf{ROTX} (keys $kset$)
\STATE \hspace{3mm} send $\langle \textsc{ROTX} \ kset, DSV_c, DS_c \rangle$ to server
\STATE \hspace{3mm} receive $\langle \textsc{ROTXReply} \ vset, dsv, ds \rangle$
\STATE \hspace{3mm} \label{line:rotx_update_dsv} $DSV_c \leftarrow max (DSV_c, dsv)$
\STATE \hspace{3mm} \label{line:rotx_cliet_update_ds}$DS_c \leftarrow maxDS(DS_c, ds)$
\RETURN $vset$

\STATE \vspace{5mm} //at partition $p^m_n$ 
\STATE \textbf{Upon} receive $\langle \textsc{ROTX} \ kset, dsv, ds\rangle$ 
\STATE \hspace{3mm} \label{line:rotx_taking_max} $DSV^m_n \leftarrow max (DSV^m_n, dsv, ds)$ 
\STATE \hspace{3mm} $vset \leftarrow \emptyset$ 
\STATE \hspace{3mm} \label{line:rotx_setting_sv} $sv \leftarrow$ $DSV^m_n$
\STATE \hspace{3mm}  \textbf{for} each $k \in kset$ do  \COMMENT{\textit{In parallel}}
\STATE \hspace{6mm} send $\langle \textsc{SliceReq} \ k, sv \rangle$ to server
\STATE \hspace{6mm} receive $\langle \textsc{SliceReply} \  v, ds'\rangle$ 
\STATE \hspace{6mm} $vset \leftarrow vset \cup \{v\}$
\STATE \hspace{6mm} \label{line:rotx_return_ds}$ds \leftarrow maxDS(ds, ds')$
\STATE \hspace{3mm} \label{line:rotx_return}send $\langle \textsc{ROTXReply} \ vset,  DSV^m_n, ds \rangle$ to client

\STATE \vspace{5mm} //at partition $p^m_n$ 
\STATE \textbf{Upon} receive $\langle \textsc{SliceReq} \ k, sv \rangle$
\STATE \hspace{3mm} \label{line:ubdateDsvWithSv}$DSV^m_n[m] \leftarrow max(DSV^m_n[m], sv[m])$
\STATE \hspace{3mm} \label{line:obtain_sv}  obtain latest version $d$ from version chain of key $k$ s.t. 
for any $\langle i, h\rangle \in maxDS(d.ds, \{\langle d.sr, d.ut\rangle\})$, $h \leq sv[i]$
\STATE \hspace{3mm} \label{line:rotx_include_self_in_ds}$ds \leftarrow maxDS (d.ds, \{\langle d.sr, d.ut\rangle\})$
\STATE \hspace{3mm} send $\langle \textsc{SliceReply} \ d.v, ds\rangle$ back to server
\end{algorithmic}
}
\end{algorithm}

%% file: results.tex
\section{Experimental Results}
\label{sec:results}

We have implemented \MYPROTOCOL protocol in a distributed key-value store called MSU-DB. MSU-DB is written in Java, and it can be downloaded from \cite{msudb}. MSU-DB uses BerkeleyDB \cite{berkeleyDb} as the storage engine. For comparison purposes, we have implemented GentleRian in the same code base. 
%
%
We run all of our experiments on AWS \cite{aws} on \texttt{c3.large} instances running Ubuntu 14.04. The specification of servers is as follows: 7 ECUs, 2 vCPUs, 2.8 GHz, Intel Xeon E5-2680v2, 3.75 GiB memory, 2 x 16 GiB Storage Capacity.

First, in Section \ref{sec:putresponse}, we investigate the effect of clock skew on PUT latency. 
In Section \ref{sec:queryampl}, we evaluate the effect of this increased PUT latency along with query amplification. 
We analyze the effectiveness of \MYPROTOCOL in reducing update visibility latency by analysis of a typical collaborative application in Section \ref{sec:updatevis}. 
Then, we evaluate the overhead of \MYPROTOCOL by comparing the throughput of \MYPROTOCOL and GentleRain in Section \ref{sec:throughput} in cases where clocks are perfectly synchronized. Finally, we evaluate the efficiency of \name ROTX protocol in Section \ref{sec:rotxResults}. 
%
%
Since our protocol is based on GentleRain, we provide a detailed comparison of our protocol with GentleRain in this section. 
We compare our protocol with other protocols in the literature in Section \ref{sec:related}.

\subsection{Response Time of PUT Operations}
\label{sec:putresponse}


To study the effect of clock skew on the response time accurately, we need to have a precise clock skew between servers. However, the clock skew between two different machines depends on many factors out of our control. To have a more accurate experiment, we consolidate two virtual servers on a single machine and impose an artificial clock skew between them.  Then, we change the value of the clock skew and observe its effect on the response time for PUT operations. A client sends PUT requests to the servers in a round robin fashion. Since the physical clock of one server is behind the physical clock of the other server, half of the PUT operations will be delayed by the GentleRain. On the other hand, \MYPROTOCOL does not delay any PUT operation and processes them immediately.  We compute the average response time for  PUT operations with value size 1K. Figure \ref{fig:ck_on_put}-(a) shows that average response time for PUT operation in GentleRain grows as the clock skew grows, while the average response time in \MYPROTOCOL is independent of clock skew.

\begin{figure}[htp]
  \centering
  \subfigure[]{\includegraphics[scale=0.37]{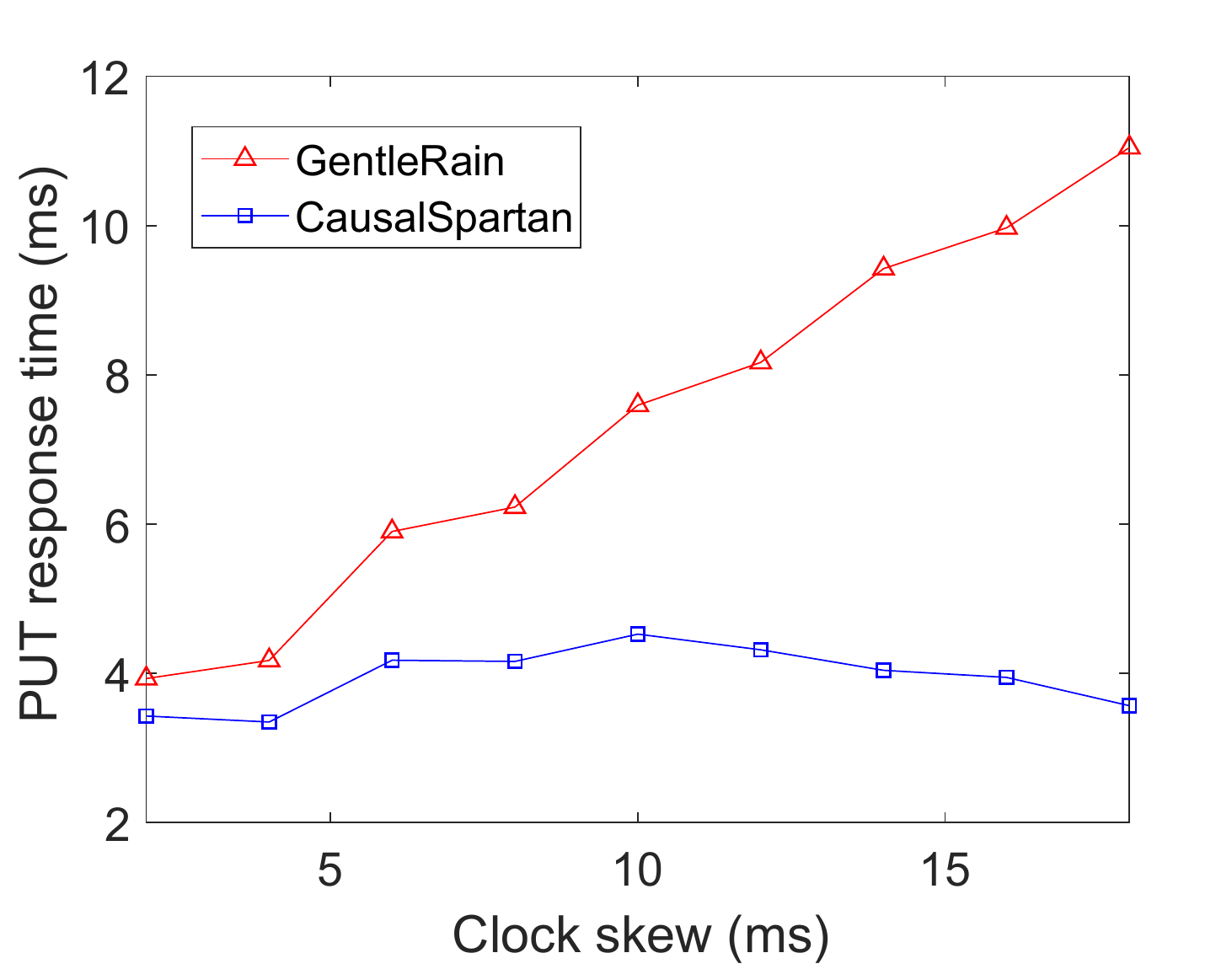}}\quad
  \subfigure[]{\includegraphics[scale=0.37]{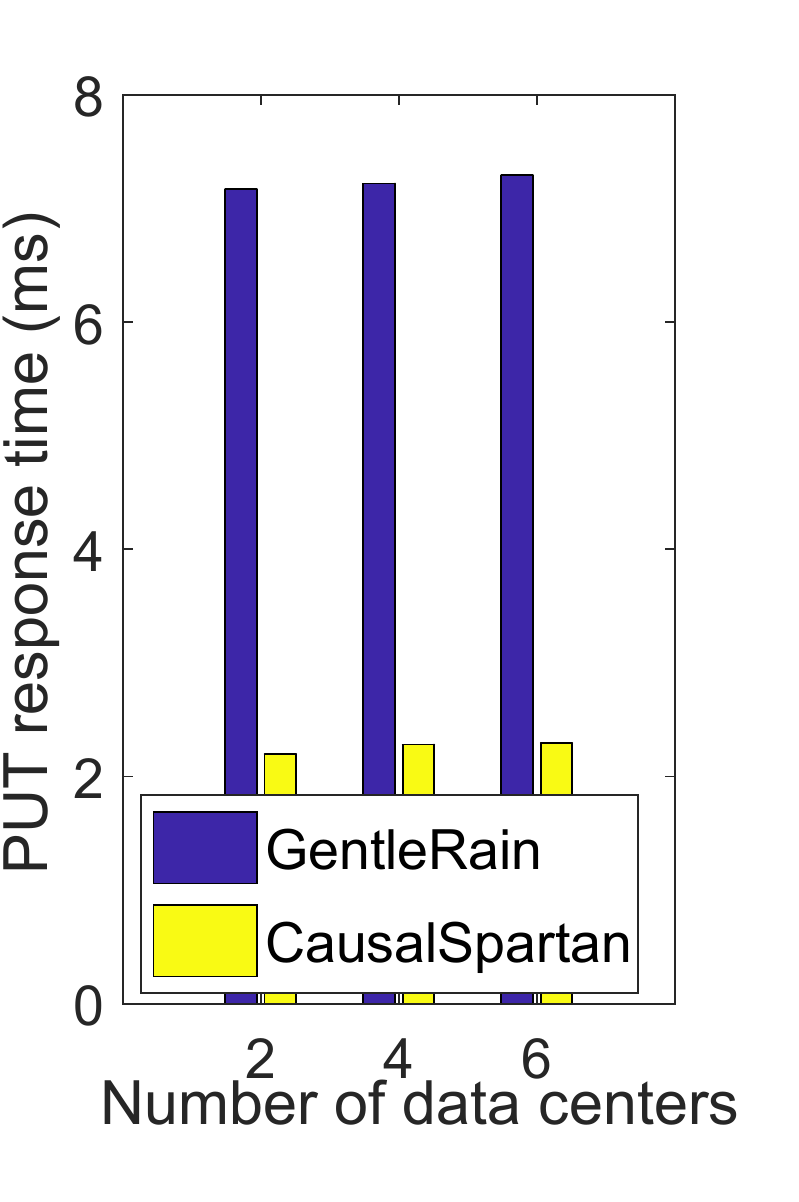}}
 \caption{The effect of clock skew on PUT response time: a) with accurate artificial clock skew when servers are running on the same physical machine b) without any artificial clock skew when servers are running on different physical machines synchronized with NTP.}
 \label{fig:ck_on_put}
\end{figure}

Next, we do the same experiment when the servers are running on two different machines without introducing any artificial clock skew.
 We run NTP \cite{ntp} on servers to synchronize physical clocks. 
 In other words, this simulates the exact condition that is expected to happen in an ideal scenario where we have two partitions within the same physical location.
 In this setting,  the client sends PUT requests to these servers in a round robin manner. 
 Figure \ref{fig:ck_on_put}-(b) shows average delay of PUT operations in GentleRain and \MYPROTOCOL. We observe that in this case, the effect of PUT latency is visible even though the servers are physically collocated and have clocks that are synchronized with NTP.

\subsection{Query Amplification}
\label{sec:queryampl}

In this section, we want to evaluate the effectiveness of \MYPROTOCOL with query amplification.
 As explained in Section \ref{sec:query}, a single user request can generate many internal queries. We define \textit{query amplification factor}  as the number of internal queries that are generated for a single request. In this section, unlike the previous section where we computed average response time for the queries, we compute the average response time for \textit{requests} each of which contains several queries specified by the query amplification factor.

Now, we want to study how the average response time changes as the query amplification factor changes.
 We simulate the scenario where the user sends requests to a web server, and each request generates multiple internal PUT operations. The web server sends PUT operations to partitions in a round robin fashion. The user request is satisfied once all  PUT operations are done. 
 We compute the average response time for different query amplification factors.

 Figure \ref{fig:ck_on_qAmp} shows average response time versus query amplification factor when we have two partitions for different clock skews. As query amplification factor increases, the response time in both GentleRain and \MYPROTOCOL increases. This is expected since each request now contains more work to be done. However, the rate of growth in \MYPROTOCOL is significantly smaller. 
 For example, for only 2 (ms) clock skew, the response time of a request with amplification factor 100 in GentleRain is 4 times higher than that in \name. Note that in practical systems higher clock skews are possible \cite{ck100, facebookClockSkew}. For example, clock skew up to $100$ (ms) is possible when the underlying network suffers from asymmetric links \cite{ck100}. In this case, for a query amplification factor of 100, the response time of GentleRain is 35 times higher than \name.

For results shown in Figure \ref{fig:ck_on_qAmp}, we used a controlled artificial clock skew to study the effect of clock skew accurately. Figure \ref{fig:ck_on_aAmp_noSim}-(a) shows the effect of amplification factor on request response time when there is no artificial clock skew, and servers are synchronized with NTP. It shows how real clock skew between synchronized servers that use NTP affects request response time. For instance, for query amplification factor 100, our experiments show that the response time of GentleRain is 3.89 times higher than that of \name. Figure \ref{fig:ck_on_aAmp_noSim}-(b) also shows the client request throughput in GentleRain and \name for different query amplification factor. 

\begin{figure}
\begin{center}
\includegraphics[scale=0.27]{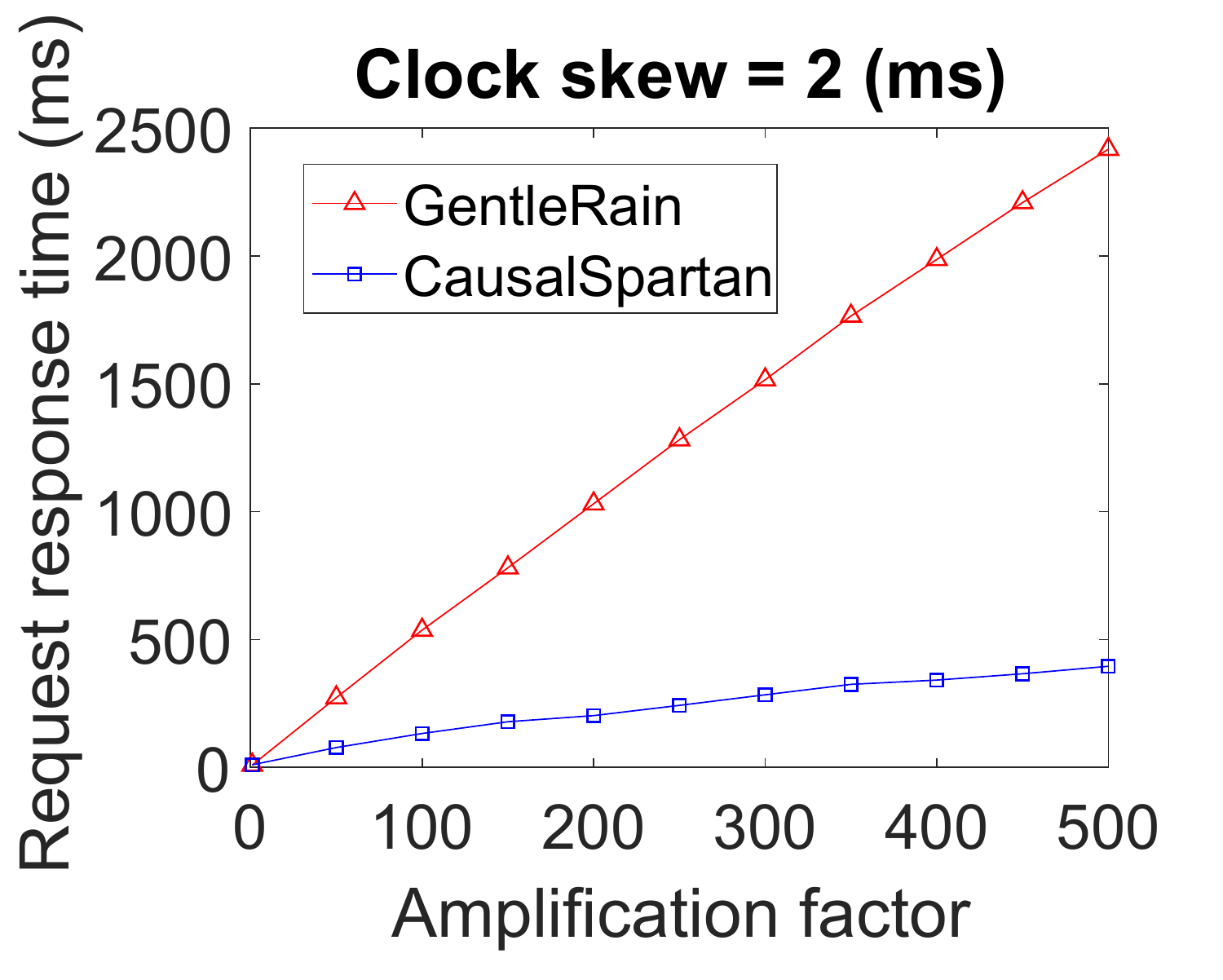}
\includegraphics[scale=0.27]{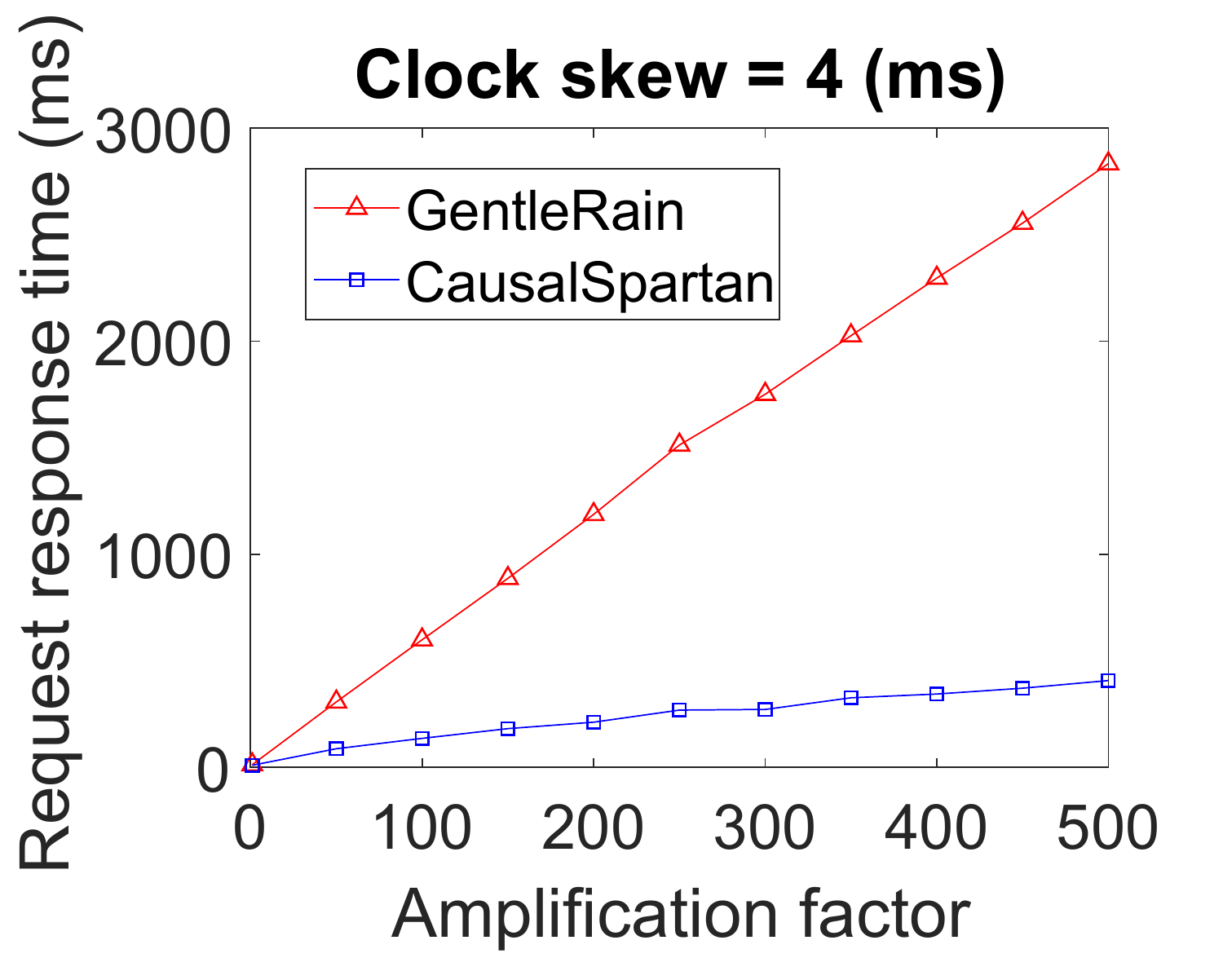}
\includegraphics[scale=0.27]{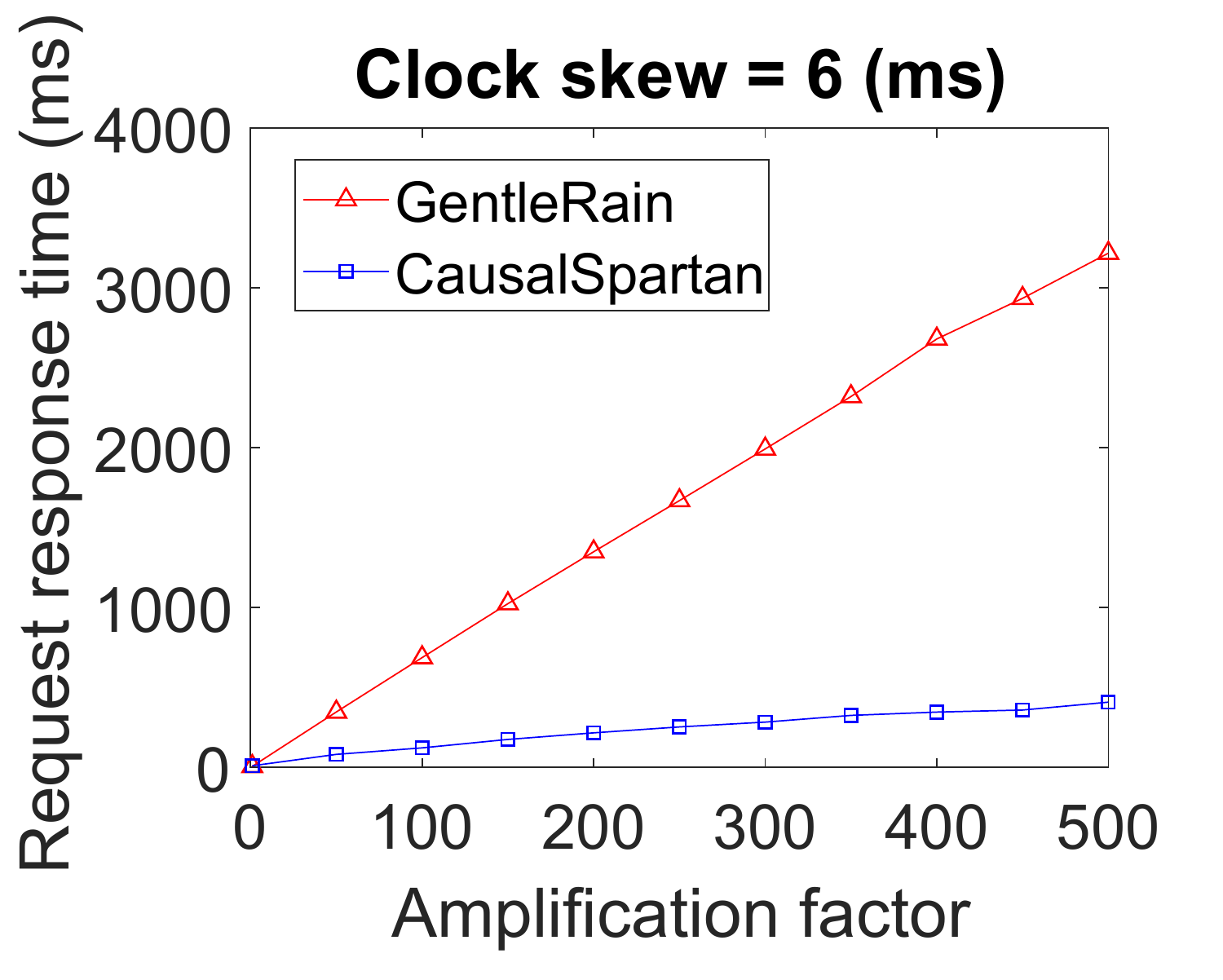}
\includegraphics[scale=0.27]{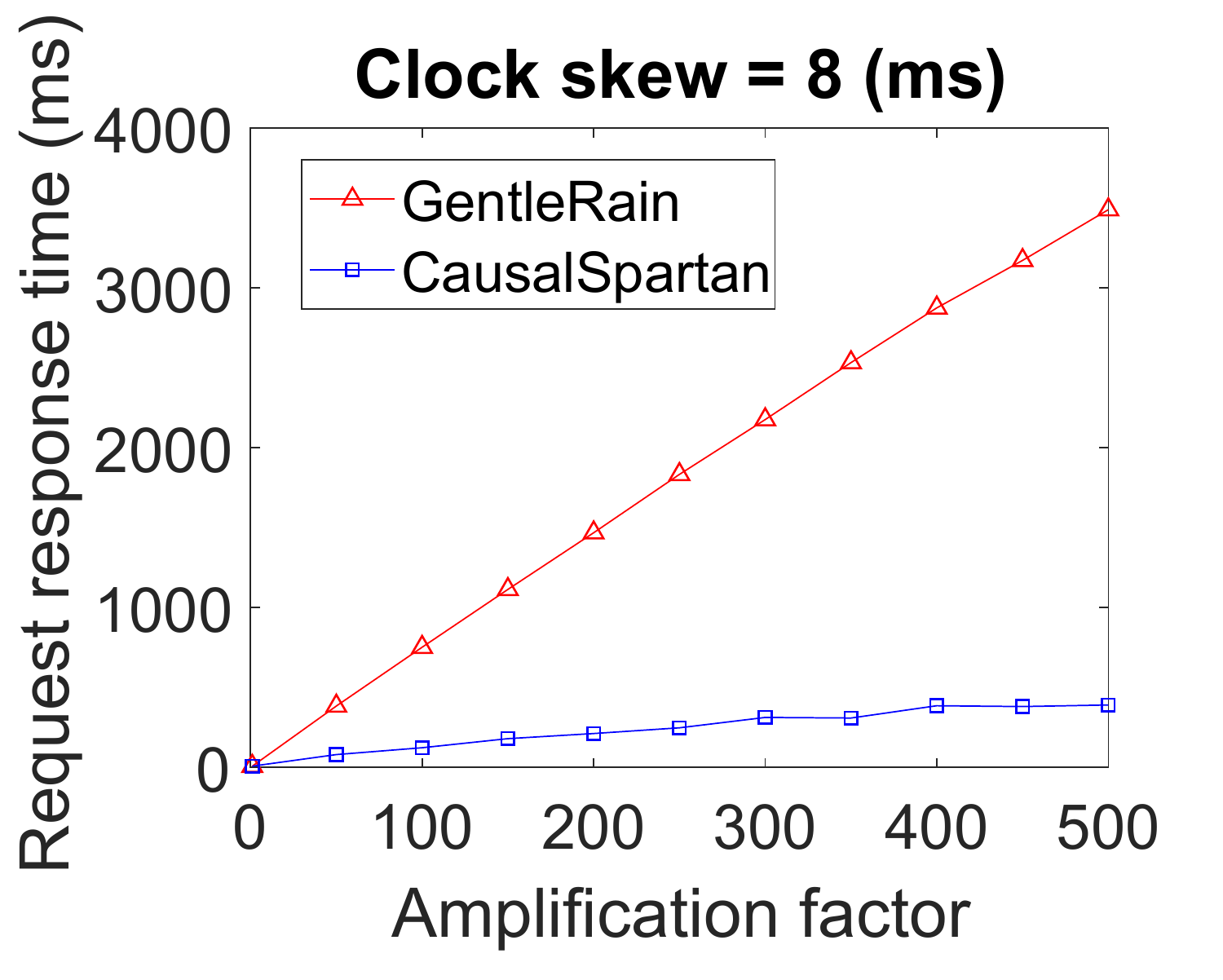}
\caption{The effect of different values of clock skew on request response time for different query amplification factor in GentleRain and \name.}
\label{fig:ck_on_qAmp}

\end{center}
\end{figure}

\begin{figure}[htp]
  \centering
  \subfigure[Request response time]{\includegraphics[scale=0.28]{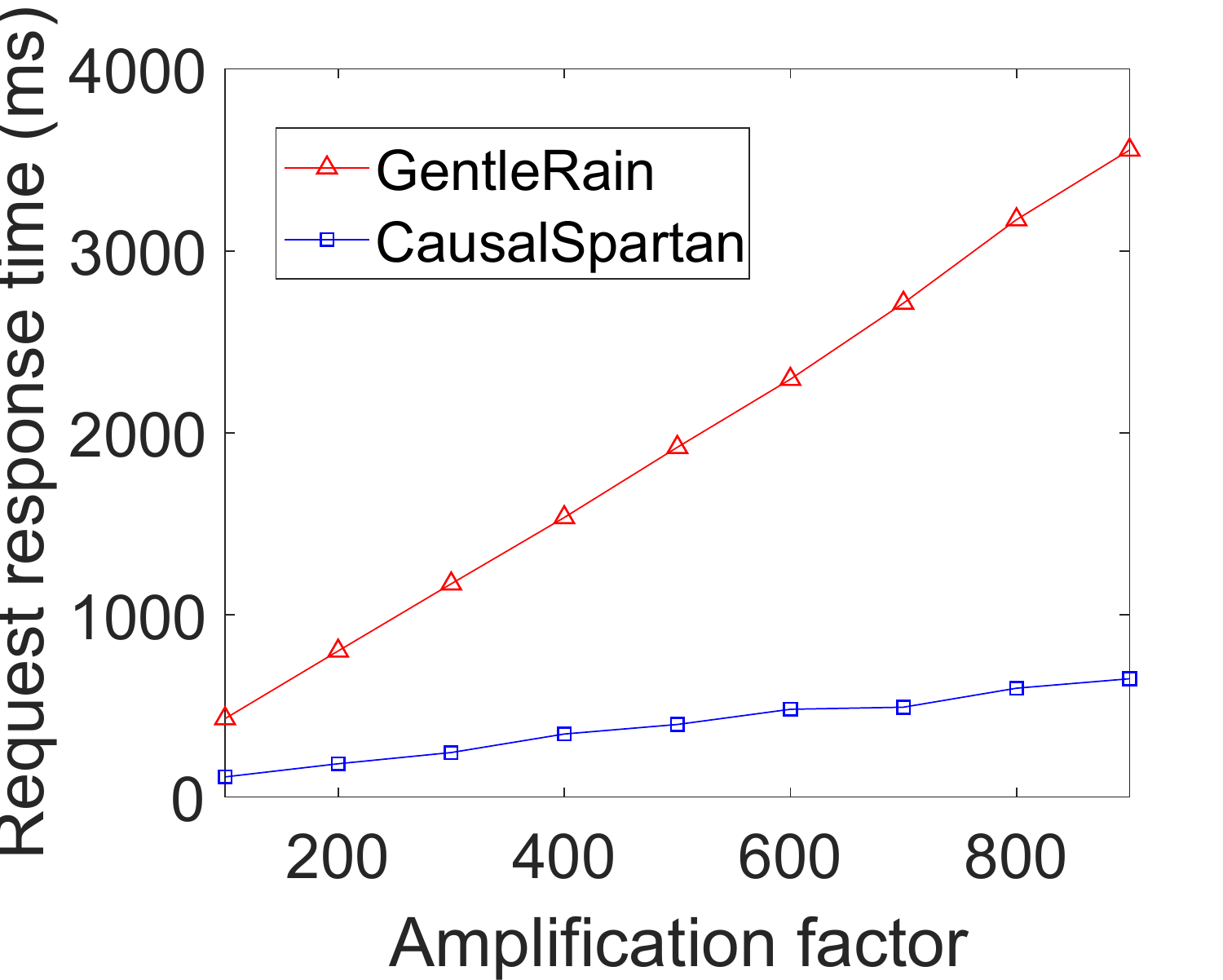}}\quad
  \subfigure[Request throughput]{\includegraphics[scale=0.28]{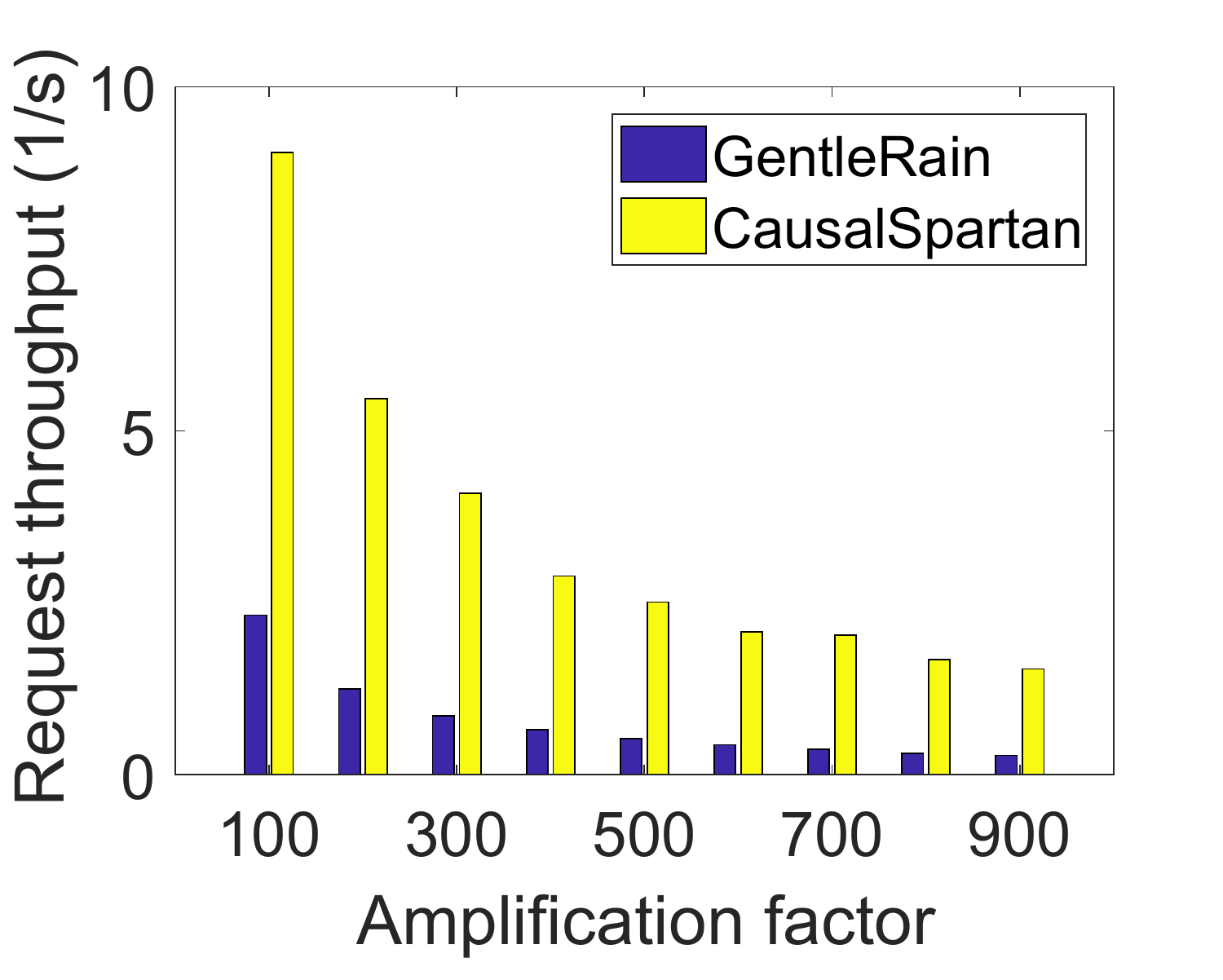}}
 \caption{The effect of amplification factor on client request response time and throughput when we have 8 partitions and 6 data centers, and all partitions are synchronized by NTP without any artificial clock skew.}
 \label{fig:ck_on_aAmp_noSim}
\end{figure}

\subsection{Update Visibility Latency}
\label{sec:updatevis}

In this section, we want to focus on update visibility latency which is another important aspect of a distributed data store. Update visibility latency is the delay before an update becomes visible in a remote replica. Update visibility latency is ultimately important for today's cloud services, as even few milliseconds matters for many businesses \cite{facebook}. 
In GentleRain, only one slow replica adversely affects the whole communication in the system by increasing the update visibility latency. In \name, we use a vector (DSV) with one entry for each data center instead of a single scalar (GST) as used in GentleRain. As a result, a long network latency of a data center only affects the communication with that specific data center and does not affect independent communication between other data centers. 

To investigate how \name \ performs better than GentleRain regarding update visibility latency, we do the following experiment: We run a data store consisting of three data centers $A$, $B$, and $C$. Client $c_1$ at data center $A$ communicates with client $c_2$ at data center $B$ via key $k$ as follows: client $c_1$ keeps reading the value of key $k$ and increments it whenever finds it an odd number. Similarly, client $c_2$ keeps reading the value of key $k$, and increments it whenever finds it an even number. The locations of data centers $A$ and $B$ are fixed, and they are both in California. We change the location of data center $C$ to see how its location affects the communication between $c_1$ and $c_2$. Table \ref{tab:ping_times} shows the round trip times for different locations of data center $C$.

\begin{table} [h]
\vspace*{-2mm}
\begin{center}
\caption{Round trip times.}
\label{tab:ping_times}
\begin{tabular}{ |l|l|l| } 
\hline
 Location of $C$ & RTT to  $A$ (ms)& RTT to  $B$ (ms)\\
\hline
California  & 1.1709114 & 0.3201521    \\
Oregon  & 21.8699663 &    20.6107391\\
Virginia  & 67.0469505 &61.2305881    \\
Ireland  & 138.2809544 &139.3212938    \\
Sydney  & 159.0899451 &    158.4004238\\
Singapore  &  175.6392972&    175.6030464\\
\hline
\end{tabular}
\end{center}
\vspace*{-4mm}
\end{table}

We measure the update visibility latency as the time elapsed between writing a new update by a client and reading it by another client. We use the timestamp of updates to compute the update visibility latency. Thus, because of clock skew, the values we compute are an estimation of actual update visibility latency.  {Figure \ref{fig:uvl}-(a) shows the update visibility latency is lower in \name than that in GentleRain. Also, the update visibility latency in GentleRain increases as the network delay between data center $C$ and $A$/$B$ increases. For example, when data center $C$ is in Oregon the update visibility latency in \name is 83\% lower than that in GentleRain. This value increases to 92\% when data center $C$ is in Singapore}.  Figure \ref{fig:uvl}-(b) shows the throughput of communication between clients as the number of updates by a client per second. The location of data center $C$ affects the throughput of GentleRain, while the throughput of \name is unaffected.

\begin{figure}[htp]
  \centering
  \subfigure[Update visibility latency]{\includegraphics[scale=0.28]{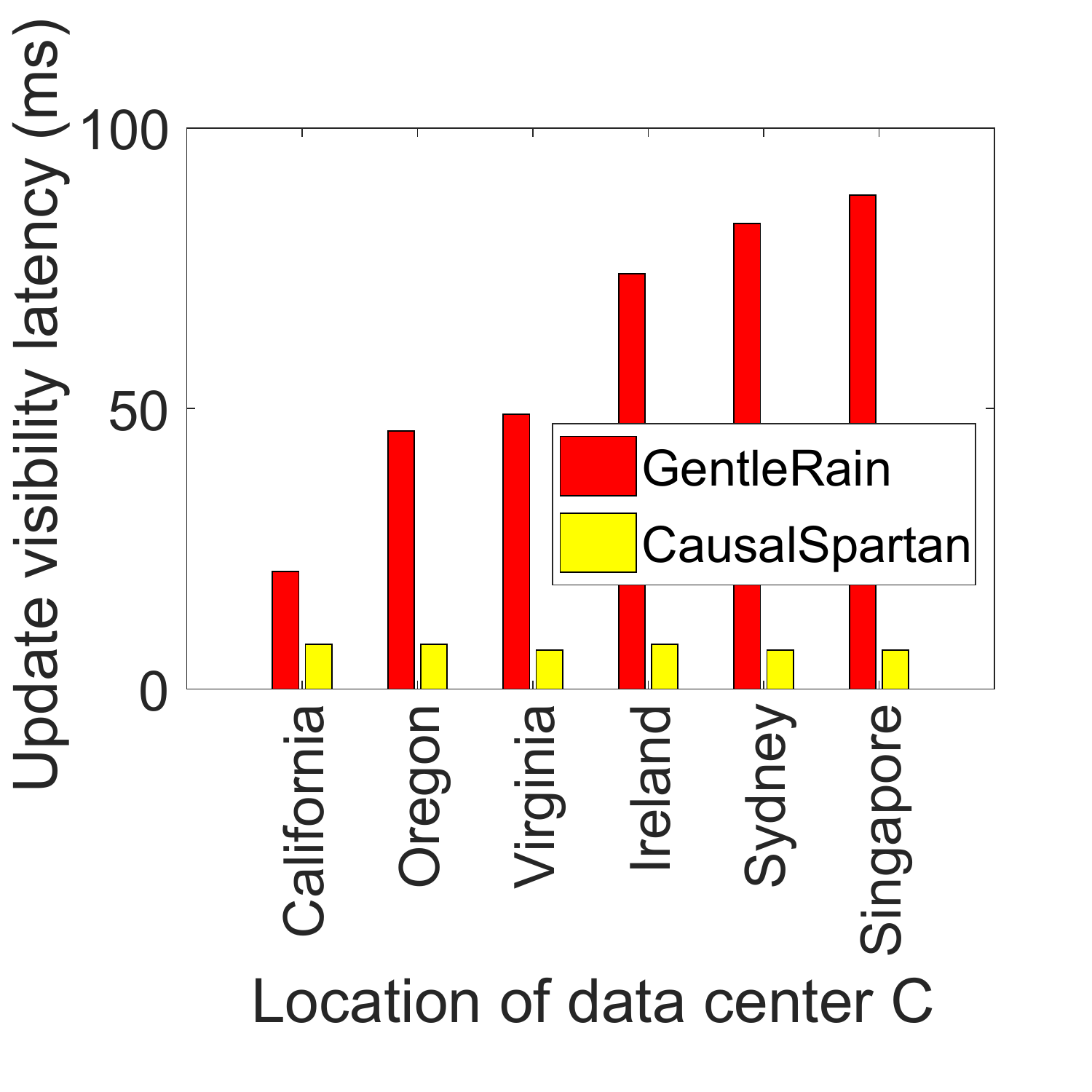}}\quad
  \subfigure[Update throughput]{\includegraphics[scale=0.28]{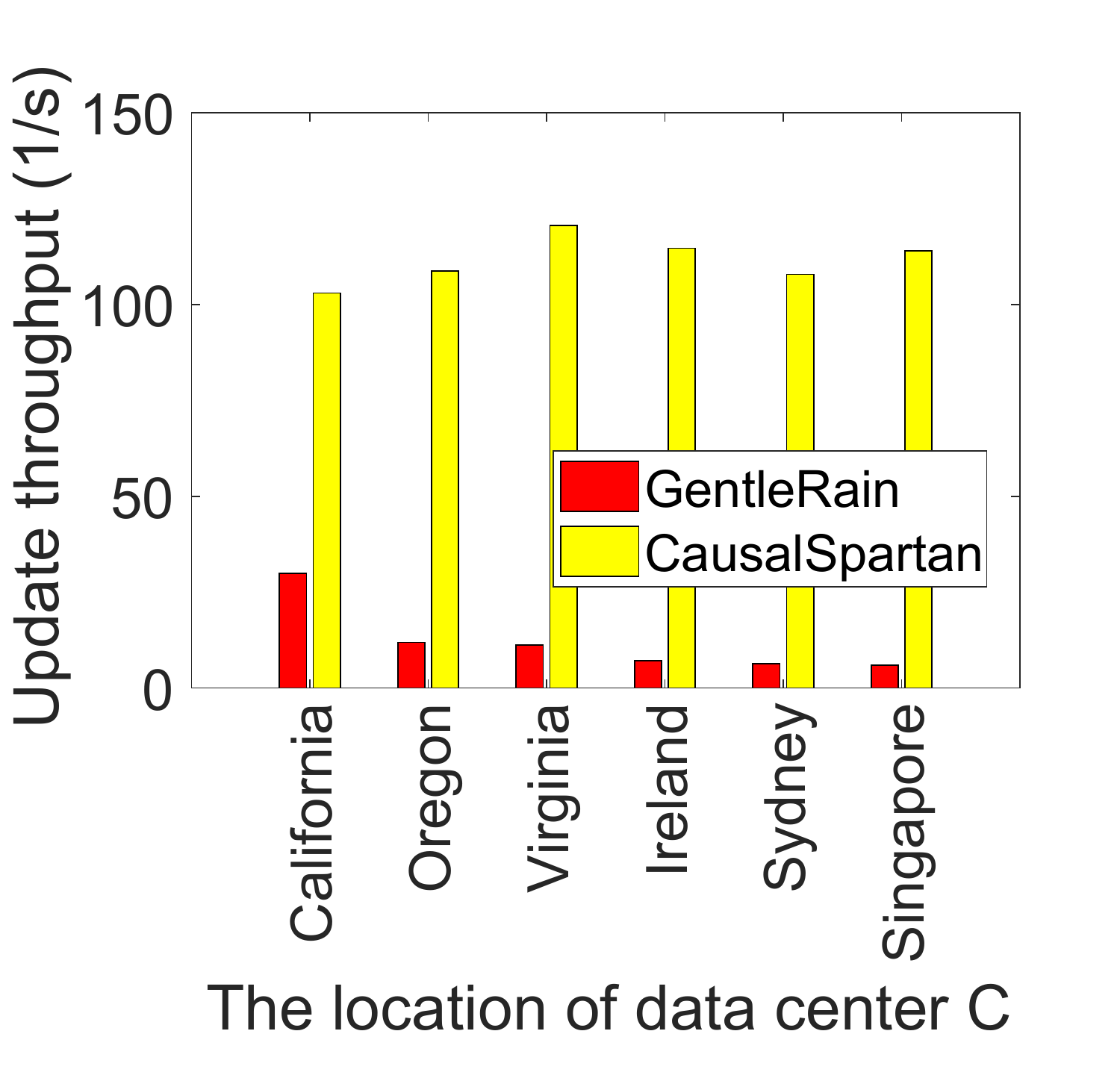}}
 \caption{How the location of an irrelevant data center adversely affects a collaborative communication in GentleRain, while \name is unaffected.}
 \label{fig:uvl}
\end{figure}

\subsection{Throughput Analysis and Overhead of \MYPROTOCOL}
\label{sec:throughput}

\MYPROTOCOL utilizes HLC to eliminates the PUT latency and utilizes DSV to improve update visibility latency. In this section, we analyze the overhead of these features in the absence of clock skew, query amplification or the collaborative nature of the application. In particular, we analyze the throughput when GET/PUT operations by the client are unrelated to each other. 

Since the two features of \MYPROTOCOL, the use of HLC and the use of DSV are independent, we analyze the throughput with just the use of HLC and with both features. Figure \ref{fig:throughput} demonstrates the throughput of GET and PUT operations. We observe that when GET/PUT operations are independent, then throughout of \MYPROTOCOL is 5\% lower than GentleRain. However, the throughput of \MYPROTOCOL with just HLC (and not DSV) is virtually identical to that of GentleRain. We note that additional experiments comparing \MYPROTOCOL with just using HLC are available in \cite{tech}. They show that just using HLC does not add to the overhead of \MYPROTOCOL.

Even though there is a small overhead of \MYPROTOCOL when GET/PUT operations are unrelated, we observe that with query amplification (that causes some  PUT operations to be delayed in GentleRain), the request throughput of \MYPROTOCOL is higher (cf. Figure \ref{fig:ck_on_aAmp_noSim}). Thus, while the throughput of basic PUT/GET operations is slightly higher in GentleRain, the throughput of actual requests issued by the end users is expected to be higher in \name.

\begin{figure}[htp]
  \centering
  \subfigure[PUT]{\includegraphics[scale=0.28]{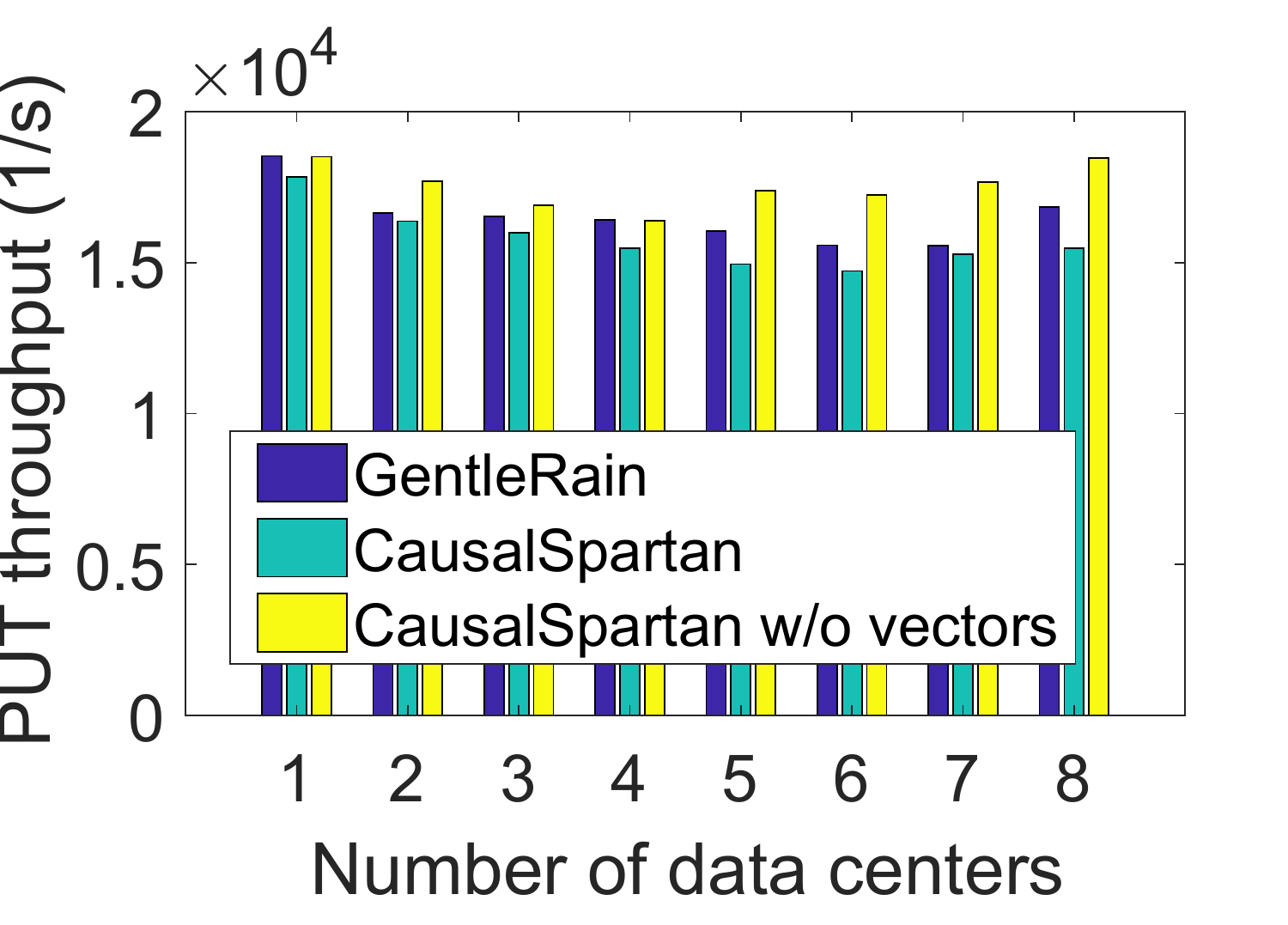}}\quad
  \subfigure[GET]{\includegraphics[scale=0.28]{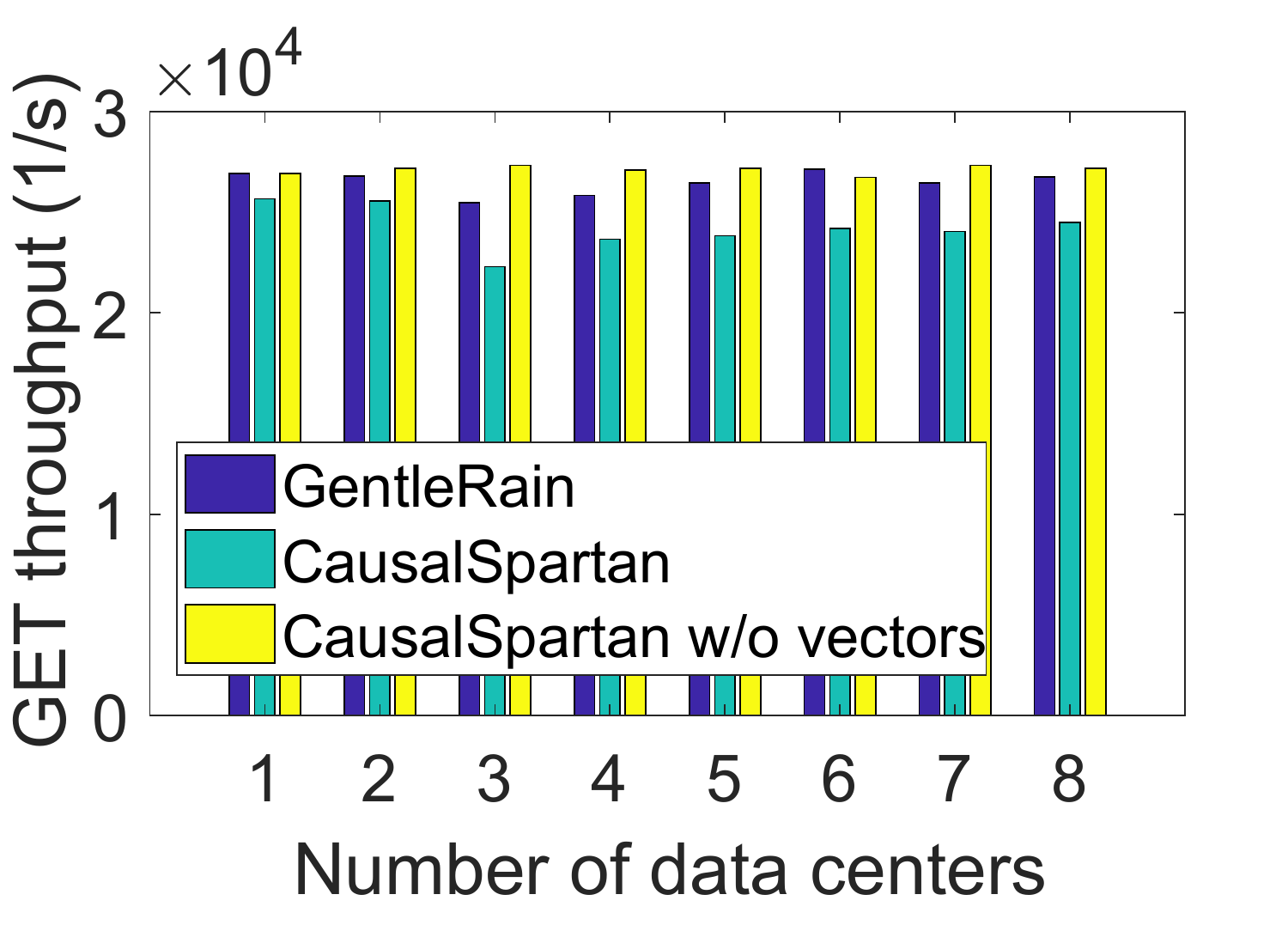}}
 \caption{The basic PUT/GET operations throughput in GentleRain and \name.}
 \label{fig:throughput}
\end{figure}

\subsection{Performance of ROTX operations}
\label{sec:rotxResults}

In this section, we compare the performance of ROTX algorithm provided in Section \ref{sec:rotx} with ROTX (named GET-ROTX in \cite{gentleRain}) of GentleRain \cite{gentleRain}. 
The sketch of GentleRain's ROTX is as follows: like PUT operations, the client includes its $dt$ with its GET-ROTX operations and sends it to one of the servers as the coordinator. Upon receiving an ROTX operation, if $|dt - GST|$ is smaller than a threshold, sever waits for $GST$ to goes higher than $dt$. The server sets the snapshot time as its $GST$ and sends requests to all partitions hosting some of the requested keys. In its request, the server includes the snapshot time. The receiving partitions return versions with timestamps smaller than the snapshot time. If  $|dt - GST|$ is higher than the threshold, the server runs the Eiger \cite{eiger}. Since (1) \cite{gentleRain} reports that this backup option was never triggered (2) there is a significant increase in metadata  and (3) there is up to two more rounds of communication between the client and servers if one intends to use \cite{eiger} for backup, in our experiments, we use threshold values that guarantee the that the backup option was not necessary.

In Section \ref{sec:updatevis}, we saw how waiting for GST can increase the update visibility latency. In the case of read-only transactions, waiting for GST in GentleRain also increases the response time. Specifically, since GentleRain blocks read-only transactions until GST is high enough, any delay in GST leads to higher response times for the clients and reduced throughput in the system. Latency in GST can be caused by slow replicas (as we saw in Section \ref{sec:updatevis}), slow partitions inside the local replica, or any other delay in GST calculation.

To evaluate the response time of transactions in the presence of a slow partition, we set up the experiment as follows: 
(1) We consider a set of hot keys that are constantly written by several clients, (2) We have some clients that perform GET with 50\% probability and ROTX with 50\% probability on hot keys. 
(3) We simulate one partition to be slow by intentionally delaying sending messages from it. Specifically, whenever the slow partition wants to send a message to other servers or the clients, we schedule it to be sent after a certain amount of time. 

Figure \ref{fig:rotxSlowPartition}-(a) shows the latency of 200 ROTX operations with size 3 (i.e. each ROTX reads 3 keys) for GentleRain in a data center with 6 partitions where one  partition is slowed by 100 (ms). The circles show the response time of ROTX operations that do not involve any key hosted by the slow partition. Diamonds, on the other hand, show the response time of the ROTX operations involving a key hosted by the slow partition. Figure \ref{fig:rotxSlowPartition}-(b) shows the same experiment results for the \name. As it is clear from the plots, in \name, ROTX operations that do not involve the slow partition are not affected by the 100 (ms) slowdown. On the other hand, in GentleRain, some of the ROTX operations that do not involve the slow partition are also affected by the slowdown. The response time of some of ROTX operations of GentelRain involving the slow partition is also significantly higher than that of \name. For 100 (ms) slowdown, the average latency of ROTX operations not involving the slow partition is 29 (ms) for GentelRain while it is 8 (ms) for \name. For ROTX operations involving the slow partition, the response time is 169 and 142 for GentleRain and \name, respectively. The improvement of \name is more significant regarding the \textit{tail-latency} \cite{book}. Figure \ref{fig:rotxSlowPartition}-(c) shows the empirical CDFs for both cases of ROTX operations for GentleRain and \name. As it is marked in Figure \ref{fig:rotxSlowPartition}-(c), for 100 (ms) slowdown, the 90th percentile of ROTX operations not involving the slow partition is 129 (ms) for GentelRain while it is only 18 (ms) for \name (i.e., 86.04\% improvement). For ROTX operations involving the slow partition, the 90th percentile response time is 289 (ms) and   203 for GentleRain and \name, respectively (i.e. 29.7\% improvement). Figure \ref{fig:rotxSlowPartition} also shows results for the case with 500 (ms) slowdown. The improvement of \name is more significant for higher slowdowns. For 500 (ms) slowdown, 90th percentile improvement goes to 93.88\%. We note that for higher percentiles also, \name shows clear improvement. For instance, for the 99th percentile, \name leads to 95.99\% and 37.70\% for not involving and involving cases, respectively, compared with GentleRain \cite{gentleRain}. Note that tail latency is very important, as it directly affects the users' experience. Although it affects a small group of clients (e.g. 1 percent of clients in case of 99th percentile), this small group are usually the most valuable users that perform most of the requests in the system \cite{book}. For this reason, many companies such as Amazon describe response time requirements for their services with 99.9th percentile \cite{book}. 

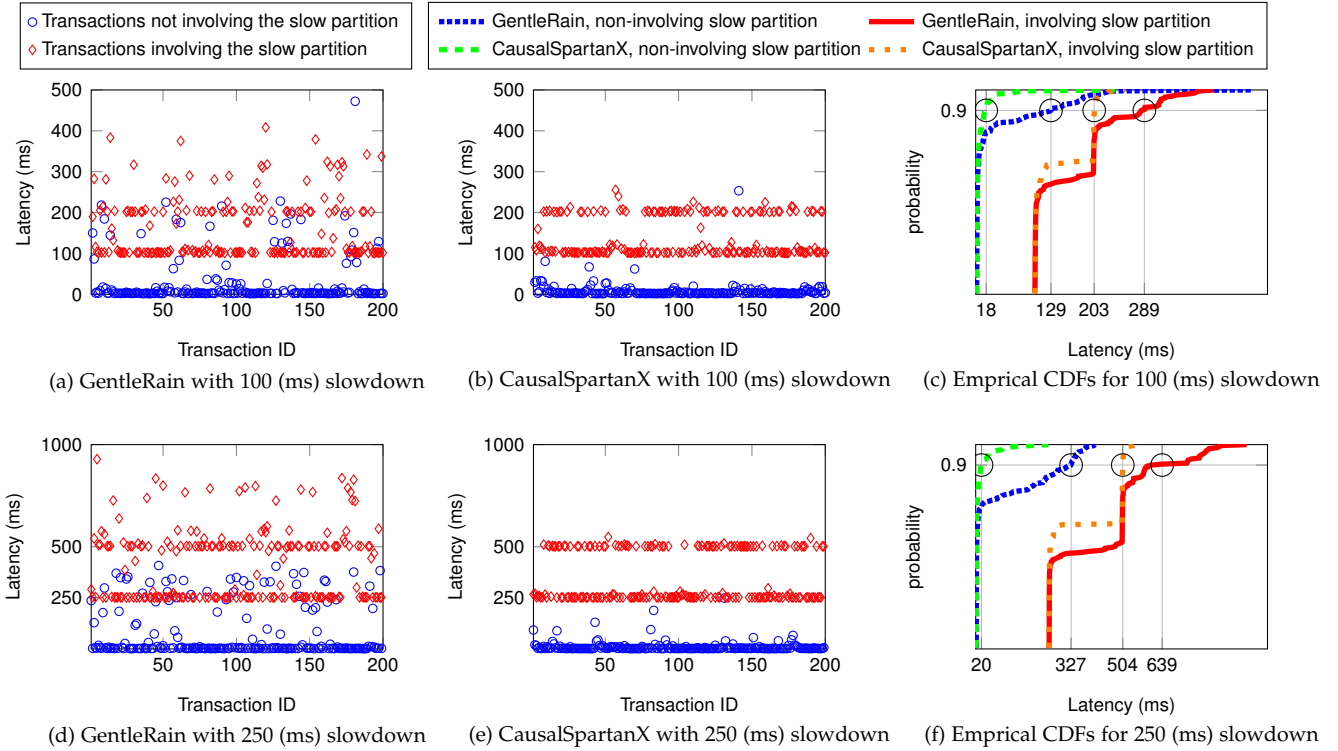
\begin{figure*}
  \centering
  
\begin{tikzpicture}
\begin{groupplot}[
group style={group size=3 by 3, 
horizontal sep= 2cm, 
vertical sep= 2cm, 
group name=myplot},
enlargelimits=false,
ytick distance=100,
xmin=1,xmax=200,
title style = {at={(0.5,-0.6)}, font=\footnotesize}
]

\nextgroupplot [title = (a) GentleRain with 100 (ms) slowdown, xlabel=Transaction ID , ylabel= Latency (ms), legend pos= south east, legend style = {at = {(1.1,1.1)}}, 
ymin=0,ymax=500]
\addplot%
    table[x=tran, y=G_L_F, col sep=comma] {data/delay_100.txt};\label{plots:G_L_F_100}
    \addlegendentry{Transactions not involving the slow partition};    
    \addlegendentry{Transactions involving the slow partition};

\addplot%
    table[x=tran, y=G_L_T, col sep=comma] {data/delay_100.txt};\label{plots:G_L_T_100}

\nextgroupplot [title = (b) \name with 100 (ms) slowdown,  xlabel=Transaction ID, ylabel= Latency (ms), legend pos= south east, 
ymin=0,ymax=500]
\addplot%
    table[x=tran, y=C_L_F, col sep=comma] {data/delay_100.txt};
\addplot%
    table[x=tran, y=C_L_T, col sep=comma] {data/delay_100.txt};

\nextgroupplot [title = (c) Emprical CDFs for 100 (ms) slowdown,  xlabel=Latency (ms), ylabel= probability, legend pos= south east,
grid=both,
legend columns=2,
legend style = {at = {(1,1.1)}}, 
ytick = {0.9},
xtick = {129, 289,    18,    203}, 
ymin=0,ymax=1,
xmin=0,xmax=500]
\addplot[blue, densely dotted, line width=2pt]%
    table[x=G_L_F, y=p, col sep=comma] {data/cdf_100.csv};
\addplot[red, solid, line width=2pt]%
    table[x=G_L_T, y=p, col sep=comma] {data/cdf_100.csv};
\addplot[green, dashed,  line width=2pt]%
    table[x=C_L_F, y=p, col sep=comma] {data/cdf_100.csv};
\addplot[orange, loosely dotted, line width=2pt]%
    table[x=C_L_T, y=p, col sep=comma] {data/cdf_100.csv};
\addlegendentry{GentleRain, non-involving slow partition};    
\addlegendentry{GentleRain, involving slow partition};  
\addlegendentry{\name, non-involving slow partition};  
\addlegendentry{\name, involving slow partition}; 
\draw(axis cs:129,0.9) circle[red, radius=1.5mm];
\draw(axis cs:289,0.9) circle[red, radius=1.5mm];
\draw(axis cs:18,0.9) circle[red, radius=1.5mm];
\draw(axis cs:203,0.9) circle[red, radius=1.5mm];
    
\nextgroupplot [title = (d) GentleRain with 250 (ms) slowdown, xlabel=Transaction ID , ylabel= Latency (ms), ytick = {250, 500, 1000},
ymin=0,ymax=1000]
\addplot%
    table[x=tran, y=G_L_F, col sep=comma] {data/delay_250.txt};\label{plots:G_L_F_250}
   
\addplot%
    table[x=tran, y=G_L_T, col sep=comma] {data/delay_250.txt};\label{plots:G_L_T_250}

\nextgroupplot [title = (e) \name with 250 (ms) slowdown,  xlabel=Transaction ID, ylabel= Latency (ms), legend pos= south east, ytick = {250, 500, 1000},
ymin=0,ymax=1000]
\addplot%
    table[x=tran, y=C_L_F, col sep=comma] {data/delay_250.txt};
\addplot%
    table[x=tran, y=C_L_T, col sep=comma] {data/delay_250.txt};

\nextgroupplot [title = (f) Emprical CDFs for 250 (ms) slowdown,  xlabel=Latency (ms), ylabel= probability, legend pos= south east,
grid=both,
legend columns=2,
legend style = {at = {(1,1.1)}}, 
ytick = {0.9},
xtick = {327,    639,    20,    504}, 
ymin=0,ymax=1,
xmin=0,xmax=1000]
\addplot[blue, densely dotted, line width=2pt]%
    table[x=G_L_F, y=p, col sep=comma] {data/cdf_250.csv};
\addplot[red, solid, line width=2pt]%
    table[x=G_L_T, y=p, col sep=comma] {data/cdf_250.csv};
\addplot[green, dashed,  line width=2pt]%
    table[x=C_L_F, y=p, col sep=comma] {data/cdf_250.csv};
\addplot[orange, loosely dotted, line width=2pt]%
    table[x=C_L_T, y=p, col sep=comma] {data/cdf_250.csv};
\draw(axis cs:20,0.9) circle[red, radius=1.5mm];
\draw(axis cs:327,0.9) circle[red, radius=1.5mm];
\draw(axis cs:504,0.9) circle[red, radius=1.5mm];
\draw(axis cs:639,0.9) circle[red, radius=1.5mm];

\end{groupplot}
      
\end{tikzpicture}
\caption{The effect of slowdown of one partition on the latency of ROTX operations. Each ROTX operations reads 3 keys.}
\label{fig:rotxSlowPartition}
\end{figure*}

%% file: impossibility.tex
\section{Impossibility Results}
\label{sec:impossibility}
In this section, we want to focus on the sticky clients assumption (i.e., client always access only one data center) made in our protocol as well as other protocols \cite{cops, eiger, orbe, gentleRain}. In Section \ref{sec:formalCausalConsistency}, we first formalize the requirements provided in Section \ref{sec:causalConsistency}. Having these femoral definitions is necessary to prove the impossibility result provided in Section \ref{sec:stickyProof}.

\subsection{Formalization of Requirements of Section \ref{sec:causalConsistency}}
\label{sec:formalCausalConsistency}
The definition of causal consistency depends upon the nature of operations allowed on the key-value store. It is possible that the addition of an operation violates the causal consistency provided by the key-value store. For example, if the key value store only supported GET operation and causal consistency is satisfied, adding a new operation such as ROTX can potentially violate causal consistency as the addition of ROTX creates new constraints that have to be satisfied as far as causal consistency is concerned. For this reason, the definitions in this paper are parameterized with a set of operations permitted on the key-value store. This approach allows us to make the definition generic so that if new operations (e.g., read/write transactions) are added, then the definition can be extended to them.  It also allows us to make the proofs in a modular fashion where we can prove correctness with respect to GET and ROTX separately and conclude the correctness of the system that supports both.


Next, we define the notion of visibility for an operation that captures whether a given version (or some concurrent/more recent version) is returned by the operation. 

\begin{definition} [Visibility]
\label{def:visibility_for_get}
We say version $v$ of key $k$ is visible for set of operations $O$ to client $c$ iff any operation in $O$ reading $k$ performed by client $c$ returns $v'$ such that $v' = v$ or $\neg (v \deps v')$.  
\end{definition}

 Note that in the definition above, $O$ is a set of operation types. For instance, a possible $O$ can be $\{GET\}$. In this case, visible for $O$ means, any $GET(k)$ operation performed by a client returns $v'$ such that $v' = v$ or $\neg (v \deps v')$. Instead of $\{GET\}$, $O$ could be $\{GET, ROTX\}$. Now, visible for $O$ means, any $GET(k)$ or $ROTX(K)$ that reads $k$ returns $v'$ such that $v' = v$ or $\neg (v \deps v')$. This way, we can flexibly define visibility for different sets of operations.

Using visibility, we define causal consistency as follows: 

\begin{definition} [Causal Consistency]
\label{def:causal}
Let $k_1$ and $k_2$ be any two arbitrary keys in the store. Let $v_1$ be a version of key $k_1$, and $v_2$ be a version of key $k_2$ such that $v_1 \deps v_2$. Let $O$ and $R$ be two sets of operations. We say 
the store is causally consistent for set of operations $O$ with set of reader operations $R$, if for any client $c$ conditions below hold: 
\begin{itemize}
    \item Any version written by client $c$ is visible for $O$ to  $c$.
    \item Once $c$ reads a version by one of operations of $R$, the version remains visible for $O$ to c.
    \item If $c$ that has read $v_1$ by one of operations of $R$, $v_2$ is visible for $O$ to  $c$.
\end{itemize}

\end{definition}


In the above definitions, we ignore the possibility of conflicts in writes. Conflicts occur when we have two writes on the same key such that there is no causal dependency relation between them. For example, when two clients independently write to a key without reading the update made by the other client.  

\begin{definition} [Conflict]
Let $v_1$ and $v_2$ be two versions for key $k$. We call $v_1$ and $v_2$ are conflicting iff $\neg (v_1 \deps v_2)$ and $\neg (v_2 \deps v_1)$. (i.e., none of them depends on the other.)
\end{definition}

In case of conflict, we want a function that resolves the conflict. Thus, we define conflict resolution function as $f(v_1, v_2)$ that returns one of $v_1$ and $v_2$ as the winner version. If $v_1$ and $v_2$ are not conflicting, any $f$ returns the latest version with respect to the causal dependency, i.e., if $v_1 \deps v_2$ then $f(v_1, v_2) = v_1$. Now, we define the notion of visibility that also captures conflicts: 

\begin{definition} [Visibility+]
\label{def:visibility+}
We say version $v$ of key $k$ is visible+ for set of operations $O$ for conflict resolution function $f$ to client $c$ iff any operation of $O$ performed by client $c$ reading $k$ returns $v'$ such that $v' = v$ or $v' = f(v, v')$.  
\end{definition}

\begin{definition} [Causal+ Consistency]
\label{def:causal+}
Let $k_1$ and $k_2$ be any two arbitrary keys in the store. Let $v_1$ be a version of key $k_1$, and $v_2$ be a version of key $k_2$ such that $v_1 \deps v_2$. Let $O$ and $R$ be two sets of operations. We say the store is causal+ consistent for set of operations $O$ with set of reader operations $R$ for conflict resolution function $f$ if for any client $c$ conditions below hold: 
\begin{itemize}
    \item Any version written by client $c$ is visible+ for $O$ for $f$ to $c$.
    \item Once $c$ reads a version by one of operations of $R$, the version remains visible+ for $O$ for $f$ to c.
    \item If $c$ that has read $v_1$ by one of operations of $R$, $v_2$ is also visible+ for $O$ for $f$ to $c$.
\end{itemize}
\end{definition}


To achieve the least update visibility latency for local updates, we define causal++ consistency that requires the data store to make all local updates visible to clients immediately.

\begin{definition} [Causal++ Consistency]
\label{def:causal++}
Let $O$, $R$, and $I$ be three sets of operations. A store is causal++ consistent for set of operations $O$ with set of reader operations $R$ and set of immediate-visible operations $I$ for conflict resolution function $f$ if conditions below hold: 

\begin{itemize}
\setlength{\itemsep}{0mm}
\item The store is causal+ consistent for $O$ with set of reader operations $R$ for $f$.
\item Any version $v$ written in data center $r$ is \textit{immediately} visible+ for operations of $I$ for $f$ to any client accessing $r$.
\end{itemize}
\end{definition}


\subsection{Necessity of Sticky Clients}
\label{sec:stickyProof}



Now, we show that in presence of network partitions, if we want availability and causal consistency while providing immediate visibility for local updates, the clients must be sticky. 

\begin{theorem}
In an asynchronous network (i.e., in presence of network partitions) with non-sticky clients, it is impossible to 
implement a replicated data store that guarantees following properties:

\begin{itemize}
\item Availability
\item Causal++ consistency (for $O=R=I=\{GET\}$ and any conflict resolution function $f$)
\end{itemize}
\end{theorem}

\begin{proof}
We prove this by contradiction. Assume an algorithm $A$ exists that guarantees availability and causal++ consistency for conflict resolution function $f$ in asynchronous network with non-sticky clients. We create an execution of $A$ that contradicts causal++ consistency: Assume a data store where each key is stored in at least two replicas $r$ and $r'$. Initially, the data center contains two keys $k_1$ and $k_2$ with versions $v_1^0$ and $v_2^0$. These versions are replicated on $r$ and $r'$. Next, the system executes as follows:


\begin{itemize}
\item There is a partition between $r$ and $r'$, i.e., all future messages between them will be indefinitely delayed. 
\item $c$ performs $GET(k_1)$ on replica $r$, and reads $v_1^0$.
\item $c$ performs $PUT(k_1, v_1^1)$ on replica $r$.
\item $c$ performs $PUT(k_2, v_2^1)$ on replica $r$. 
\item $c'$ performs $GET(k_2)$ on replica $r$. Since $v_2^1$ is a local update, and in causal++ consistency, local updates have to be immediately visible, the value returned is $v_2^1$.
\item $c'$ performed $GET(k_1)$ on replica $r'$. Because of the network partition, there is no way for replica $r'$ to learn $v_1^1$. Thus, the value returned is $v_1^0$.

\end{itemize}

Since $v_1^1 \deps v_1^0$, for any conflict resolution function $f$, $f(v_1^0, v_1^1) = v_1^1$. Thus, according to Definition \ref{def:visibility+}, $v_1^1$ is not visible+ for $f$ to client $c'$. It is a contradiction to causal++ consistency, as $c'$ has read $v_2^1$, but its causal consistency $v_1^1$ is not visible+ to $c'$.
\end{proof}


The necessity of sticky clients for the causal consistency has been investigated in the literature \cite{peter}. The existing proof, however, is based on the impossibility of read-your-write consistency which is part of causal consistency \cite{peter, jerzy}. In other words, since read-your-write consistency is impossible for non-sticky clients for an always-available system in presence of network partition, and because read-your-write consistency is a part of the causal consistency \cite{peter, jerzy}, the impossibility also applies to the causal consistency. However, the read-your-write consistency can be achieved even for non-sticky clients, if clients cache what they have read or written, and use this cache in the read time if a server has an older value. It is straightforward to see that our proof still holds even when clients can cache their read and writes.

%% file: related.tex
\section{Related Work}
\label{sec:related}


In this section, we review some of the previous protocols for causally consistent distributed data stores. We first review protocols for basic PUT and GET operations. Then, we focus on protocols for read-only transactions.

\subsection{Previous Work on Causal Consistency for Basic Operations}
There are several proposals for causally consistent replicated data stores such as \cite{bayou, lazyReplication, practi, causalMemory} where each replica consists of only one machine. Such an assumption is a serious limitation for the scalability of the system, as the whole data must fit in a single machine.  

To solve this scalability issue, we can partition the data inside each replica. When we have more than one machine in each replica, an important issue is how we want to make sure all causal dependencies are visible in the replica if some of them reside in other partitions. Some existing protocols such as \cite{cops, eiger} check dependencies by sending messages to other partitions every time a server receives any replicate message. This approach suffers from high message complexity. 

To eliminate the need for dependency check messages, GentleRain \cite{gentleRain} relies on physical clocks of partitions. However, as we discussed in Section \ref{sec:motivation}, clock anomalies such as clock skew among servers can adversely affect the performance of systems that rely on the synchronized physical clock. 
\name solves this problem by using HLCs instead of physical clocks. The idea of using HLC for providing causal consistency for the first time proposed in our technical report \cite{tech}. This report provides a version of \name without the DSVs and provides additional details about the effect of using HLC to provide causal consistency.

Another issue in GentleRain is that the quality of communication between replicas directly affects the update visibility latency even when a replica is irrelevant in a communication. Thus, even one slow replica adversely affects the update visibility latency in the whole system. \name solves this problem by keeping track of dependencies written in different data center separately. 

A framework for designing adaptive causal consistency protocols is provided in \cite{adaptive}. This framework allows protocol designers to easily trade off between different conflicting objectives such as lower update visibility latency and lower metadata. It also provides a basis for designing self-adaptive protocols that change themselves to match the actual usage pattern of the clients.

\subsection{Previous Work on Causal Consistency for Read-only Transactions}

COPS-GT \cite{cops} is a read-only transaction algorithm that is based on the COPS \cite{cops} algorithm that uses explicit dependency checking. COPS-GT may require two rounds of communication between the client requesting the transaction and the partitions involved in the transaction. These two rounds of communication increase the response time for the client, especially when network latency between the client and the partitions is high. Compared with COPS-GT, our algorithm requires only one round of communication between the client and one of the partitions involved in the transaction. The rest of communications are done between partitions inside the datacenter which have negligible network latency. COPS-SNOW \cite{snow} is an improvement of COPS-GT algorithm. Unlike COPS-GT, COPS-SNOW requires only one round of communication between the client and the servers. COPS-SNOW achieves this by shifting complexity from read operations to write operations. Comparing to COPS-SNOW, CausalSpartanX has significantly smaller metadata. Specifically, in COPS-SNOW, for each version $X$, we need to store a list of all transactions that have read a version older than $X$. For each key, we also need to keep track of all transactions that have read the current version. Another problem of COPS-SNOW is doing explicit dependency check for every write operation. To do this explicit dependency check, nodes need to constantly communicate with each that increase the message complexity of the algorithm. 

We explained GentleRain blocking ROTX algorithm that may require more than one round of communication with servers in Section \ref{sec:rotxResults}. Orbe \cite{orbe} and ChainReaction \cite{chainReaction} are two other blocking algorithms. Contrarian \cite{cont} is a time-based causal consistency protocol that provides non-blocking ROTX operations. However, it requires two rounds of client-server communication. Compared with \cite{cont}, \name trades one round of client-server communication with one round of server-server communication which results in a lower latency for RTOX operations when client-server communication delay is higher than server-server communication delay. For instance, if the client is located in Oregon, and servers are located in California, based on the RTTs provided in Table \ref{tab:ping_times}, \name provides 97\% improvement compared with \cite{cont} regarding the latency of ROTX operations. Even when the client is co-located with the servers, the latency of client-server communication is typically higher due to firewalls and other security checking across different security zones, because usually, the database servers are in the same zone while application servers are in a different zone.

Wren\cite{wren} relies on client cache to provide causally consistent transactions. Wren keeps track of the dependencies of a version with two scalars; one for local updates and one for remote updates. This enables Wren to eliminate the GentleRain’s blocking problem. However, since all remote dependencies are tracked by a single scalar, like GentleRain, Wren still suffers from increased update visibility in cases of slow replicas (cf. Section \ref{sec:updatevis}).  Another issue with Wren is its requirement for the clients to cache the data that they have written in the system which means more work need to be done on the client-side. 

Occult \cite{occult} is another causally consistent protocol. The main design goal of Occult is to decrease the update visibility. In fact, Occult pushes this direction to the end and makes all the updates immediately visible for the clients. However, to achieve this, Occult scarifies the always-availability. Specifically, a client may need to retries their read requests to get a causally consistent value of a key. This problem contradicts one of the most important features of the causal consistency that makes it favorable over stronger consistency models that is availability during network partitions. Another important problem of Occult is its single-leader replication policy. Specifically, all writes in the system must be done in a single replica. Occult uses this policy to manage its metadata overhead that is high in nature (i.e., in the order of the number of all partitions where clients can write) to avoid false positive causal consistency violation detection.

%% file: conclusion.tex
\section{Conclusion}
\label{sec:conclusion}

In this paper, we presented \MYPROTOCOL, a time-based protocol for providing causal consistency for replicated and partitioned key-value stores. Unlike existing time-based protocols such as GentleRain that relies on the physical clocks, our protocol is robust to clock anomalies thanks to utilizing HLCs instead of physical clocks. One of the important effects of eliminating this delay occurs for query processing when a single query results in multiple internal GET/PUT operations on the key-value store. \MYPROTOCOL guarantees that the response time for client operations is unaffected by clock anomalies such clock skew. For example, for the clock skew of 10 (ms), the average response time for PUT operations of \MYPROTOCOL was 4.5 (ms) whereas the average response time of GentleRain was 7.6 (ms).
Also, the correctness of \MYPROTOCOL is unaffected by NTP kinks such as leap seconds, non-monotonic clock updates and so on. 

This reduction in response time is especially important for federated data centers, virtual data centers, and multi-cloud environment. In federated data centers, the data is created by different entities, and all of the partitions may not be physically collocated. Likewise, a virtual data center can be created by utilizing available resources from different cloud service providers to minimize cost. In a multi-cloud environment, the data is split intentionally in such a way that no provider has access to the data but the actual data can be accessed only by clients that have access (with proper credentials) to all providers simultaneously. 
Multi-cloud environments are desired to avoid vendor lock-in.
A key characteristic of such data centers is that the partitions in a data center may be geographically distributed.
 In this case, one has to rely on the NTP protocol for clock synchronization. Typical NTP synchronization provides clock synchronization to be within 10 (ms). Even at this level of synchronization, \MYPROTOCOL reduces the average response time from 814 (ms) to 124 (ms) (for a query that consists of 100 operations) and from 3800 (ms) to 407 (ms) (for a query that consists of 500 operations). Moreover, NTP clock synchronization errors may be even large (e.g., 100 (ms)). In the case of 100 (ms) clock skew, \MYPROTOCOL reduces the average response time from 4540 (ms) to 124 (ms) (for a query that consists of 100 operations).


Another advantage of \MYPROTOCOL is reducing the update visibility latency. Update visibility latency can cause a substantial increase in latency for collaborative applications where two clients read each other's updates to decide what actions should be executed. As a simple application of this collaborative application, we considered the abstract bidding problem where one client reads the updates (using data center $A$) from another client (using data center $B$) and decides to increase its own bid until a limit is reached. We performed this experiment where $A$ and $B$ were in California, but the location of another data center, say $C$ was changed. \MYPROTOCOL performance remained unaffected by the location of $C$. By contrast, in GentleRain, the latency increased from 46 (ms) to 88 (ms) when we move data center $C$ from Oregon to Singapore.

In addition to the basic GET and PUT operations, \name provides read-only transaction operation, ROTX, that allows application developers to read a set of keys such that the returned values are causally consistency with each other as well as with the client past reads. \name ROTX is non-blocking, i.e., servers receiving the request can read the requested values immediately, and it only requires one round of communication between the client and the servers. More importantly, slow servers that are not involved in the transaction do not affect the response time of it. This feature provides a significant improvement for \name, especially regarding the tail latency, compared with the existing protocol such as GentleRain where the slowdown of a single partition in a datacenter affects the response time of all transactions in the datacenter. Our experiment shows that for a 100 (ms) slowdown of a partition, the 90th percentile of ROTX response time of \name shows 86.04\% and  29.7\% improvement over that of GentleRain, for transactions not involving the slow partition and transactions involving the slow partition, respectively. For higher slowdowns and higher percentiles, this improvement is more significant. For example, for 500 (ms) slowdown and 99th percentile, these numbers go up to 95.99\% and 37.70\%, respectively. 
 
Finally, we provided an impossibility result which states in presence of network partitions the locality of traffic is necessary, to have an always available causally consistent data store that immediately makes local updates visible. We note that the assumption about client accessing only the local data center is made in several works in the literature \cite{cops, eiger, orbe, gentleRain}. Our argument shows that this assumption is essential in them. This impossibility result is different than the existing impossibility result that requires sticky clients for read-your-writes \cite{peter} and still hold even when clients can cache their past read and writes.

%% file: correctness.tex

In this section, we focus on the correctness of our proposed protocol. We want to show that the data store running our protocol is a causal++ consistent for $R = O = \{GET, ROTX\}$, $I = \{GET\}$, and conflict resolution function last-write-wins. For ROTX operations, we also need to show that the set of values returned by ROTX operations are causally consistent with each other. 

In practice, in addition to the consistency, we want all data centers to eventually converge to the same data. In other words, we want an update occurred in a data center to be reflected in other connected data centers as well.  We define two data centers connected, if there is no network partition that prevents them from communication. Now, we define convergence as follows: 

\begin{definition} [Convergence] 
\label{def:convergence}
Let $v_1$ be a version for key $k$ written in data center $r$.  
\begin{itemize}
\setlength{\itemsep}{0mm}
\item Let data center $r'$ be continuously connected to data center $r$, and 
\item for any version $v_2$ such that $v_1 \deps v_2$, let data center $r'$  be continuously connected to data center $r''$ where version $v_2$ is written. 
\end{itemize}
The data store is convergent+ for set of operations $O$ for conflict resolution function $f$ if $v_1$ is eventually visible+ for $O$ for $f$ to any client accessing  $r'$.
\end{definition}


\subsection{Causal Consistency}

With the definitions of causal++ consistency defined in Section \ref{sec:impossibility}, and convergence defined in this section, it is straightforward to note the following observations: 

\begin{observation}
\label{obs:comp_R}
If store $S$ is 
\begin{itemize}
\item causal++ consistent for $O$, with set of reader operations $R_1$, and set of instantly visible operations $I$ for conflict resolution $f$, and
\item causal++ consistent for $O$, with set of reader operations $R_2$, and set of instantly visible operations $I$ for conflict resolution $f$,
\end{itemize}

then, $S$ is causal++ consistent for $O$, with set of reader operations $R_1 \cup R_2$, and set of instantly visible operations $I$ for conflict resolution $f$.
\end{observation}

\begin{observation}
\label{obs:comp_O}
If store $S$ is 
\begin{itemize}
\item causal++ consistent for $O_1$, with set of reader operations $R$, and set of instantly visible operations $I$ for conflict resolution $f$, and
\item causal++ consistent for $O_2$, with set of reader operations $R$, and set of instantly visible operations $I$ for conflict resolution $f$,
\end{itemize}

then, $S$ is causal++ consistent for $O_1 \cup O_2$, with set of reader operations $R$, and set of instantly visible operations $I$ for conflict resolution $f$.
\end{observation}

\begin{observation}
\label{obs:comp_I}
If store $S$ is 
\begin{itemize}
\item causal++ consistent for $O$, with set of reader operations $R$, and set of instantly visible operations $I_1$ for conflict resolution $f$, and
\item causal++ consistent for $O$, with set of reader operations $R$, and set of instantly visible operations $I_2$ for conflict resolution $f$,
\end{itemize}

then, $S$ is causal++ consistent for $O$, with set of reader operations $R$, and set of instantly visible operations $I_1 \cup I_2$ for conflict resolution $f$.
\end{observation}

\begin{observation}
\label{obs:comp_conv_O}
If store $S$ is 
\begin{itemize}
\item convergent for set of operations $O_1$ for conflict resolution function $f$, and 
\item convergent for set of operations $O_2$ for conflict resolution function $f$,
\end{itemize}

then, it is convergent for set of operations $O_1 \cup O_2$ for conflict resolution function $f$.
\end{observation}

\subsubsection{Causal Consistency for $GET$ Operations}
\label{sec:pp_cs}

As explained in Section \ref{sec:protocol}, we store a dependency set for each version. This set contains at most one entry per data center. Whenever a client reads a version, the client updates its dependency set by the dependency set of the version (see Line \ref{line:client_update_ds} of Algorithm \ref{alg:client}, and Line \ref{line:rotx_cliet_update_ds} of Algorithm \ref{alg:rotx}). This dependency set later will be used as the dependency set of any version written by this client. Thus, dependencies are transitive. In addition, whenever a client writes a version $v$, we update the $DS_c$ of the client in Line \ref{line:put_update_ds} of Algorithm \ref{alg:client} to capture dependency on $v$. Thus, $ds$ of a version written by client $c$ captures dependency on any other version previously written by $c$, and any version $v$ previously read by client $c$ along with all dependencies of $v$. Specifically, we have the following observation, 

\begin{observation}
\label{obs:dep_h'_h}
Let $v_1$ and $v_2$ be two versions for two keys. If $v_1 \deps v_2$, then for any member $\langle i, h\rangle \in v_2.ds$, there exists $\langle i, h'\rangle \in v_1.ds$ such that $h' \geq h$. 
\end{observation}

Whenever a server returns a version, it includes the  update time of the version in the dependency set returned with the version/versions (see Line \ref{line:includeUt} of Algorithm \ref{alg:server1}, and Line \ref{line:rotx_include_self_in_ds} of Algorithm \ref{alg:rotx}). Thus, based on Observation \ref{obs:dep_h'_h}, we have the following observation,
\begin{observation}
\label{obs:dep_h_ut}
Let $v_1$ and $v_2$ be two versions for two keys such that $v_2$ is written in data center $j$. If $v_1 \deps v_2$, then $v_1.ds$ has member $\langle j, h\rangle$ such that $h \geq v_2.ut$.   
\end{observation}

According to Line \ref{line:takeMaxforDt} of Algorithm \ref{alg:server1}, using $HLC$ algorithm, our protocol assigns a timestamp higher than the timestamps of all of the dependencies of a version. Thus, we have

\begin{observation}
\label{obs:dep_ut_ut}
Let $v_1$ and $v_2$ be two versions for two keys. If $v_1 \deps v_2$, then $v_1.ut > v_2.ut$. 
\end{observation}
We consider the last-write-wins policy for conflict resolution. The last-write-wins means, if two versions are conflicting, then the winner version is the one with higher timestamp. If timestamps are equal, the winner version is the one with higher replication id. Specifically, we have 
\begin{observation}
\label{obs:f_property}
For conflict resolution function $f=\textsf{last-writer-wins}$, $f(v,v') = v$ iff 
\begin{itemize}
\item $v \deps v'$, or
\item $\neg(v \deps v') \wedge \neg(v' \deps v) \wedge \langle v.ut, v.sr\rangle$ is lexicographically greater than $\langle v'.ut, v'.sr\rangle$.
\end{itemize}
\end{observation}

\begin{lemma}
\label{lemma:local_visibility}
All versions written in data center $i$ are immediately visible+ for $\{GET\}$ to any client at data center $i$. 
\end{lemma}

\begin{proof}
Suppose a local version $v$ is not immediately visible+ for $\{GET\}$ to a client. That requires that in response to a $GET$ operation, another version $v'$ is returned such that $f(v, v') = v$. According to Observations \ref{obs:f_property} and \ref{obs:dep_ut_ut}, $\langle v.ut, v.sr\rangle$ in any case is lexicographically higher than $\langle v'.ut, v'.sr\rangle$. Thus, since the version is local, the hosting partition has this version in the version chain of the key. Thus, according to Line \ref{line:get:obtain} of Algorithm \ref{alg:server1}, it is impossible to return $v'$ (contradiction).
\end{proof}

A non-local version is not visible for $\{GET\}$ in a data center, either if it has not arrived its hosting partition, or DSV of the hosting partition is behind of one of the dependencies. Specifically, we have the following Lemma, 

\begin{lemma}
\label{lem:not_visible_get}
Let $v$ be a version that is written in partition $p^m_n$. If $v$ is not visible+ for $\{GET\}$ for conflict resolution function $f=$ last-write-wins to a client in data center $i$, then
\begin{itemize}
\item $v$ has not arrived partition $p^i_n$, or
\item there is a member $\langle k, h\rangle \in v.ds$, such that $DSV^i_n[k] < h$. 
\end{itemize}

\end{lemma}

\begin{proof}
When $v$ is not visible+ for $\{GET\}$ in data center $i$, it means in the response of a $GET$ operation another version $v'$ is returned such that $f (v, v') = v$ (see Definition \ref{def:visibility+}). According to Observation \ref{obs:f_property}, two cases are possible. In the first case, according to Observation \ref{obs:dep_ut_ut}, $v.ut > v'.ut$, thus in both cases $\langle v.ut, v.sr \rangle > \langle v'.ut, v'.sr\rangle$. Now suppose both conditions of the this lemma are false. Then, according to Line \ref{line:get:obtain} of Algorithm \ref{alg:server1}, it is impossible that $GET$ operation returns $v'$ (contradiction).
\end{proof}


We define $DSV^i_{real}$ as entry-wise minimum of all $VV$ of partitions in data center $i$. In other, words, for all $1 < m < M$, $DSV^i_{real}[i] = \min_{1 < n < N} VV^i_n[m]$.

Now, we have the following lemma about the relation of the $DSV_{real}$ and other DSV and $DS$ values inside a data center: 

\begin{lemma}
\label{lem:real_is_greater}
In data center $i$, 
\begin{itemize}
\item for all $1 < n < N$, and $1 < m < M \wedge i \neq m$, $DSV^i_n[m] \leq DSV^i_{real}[m]$, and
\item for any client $c$ accessing data center $i$, for all $1 < m < M \wedge i \neq m$, $DSV_c[m] \leq DSV^m_{real}[m]$, and 
\item for any client $c$ accessing data center $i$, for any $\langle k, h
 \rangle \in DS_c \wedge k \neq i$, $h \leq DSV^i_{real}[m]$.
\end{itemize}
\end{lemma} 

\begin{proof}
We prove this lemma, by induction over time.

At the beginning all $VV$, $DSV$, and $DS$ values at servers and client are zero. Thus, $DSV^m_{real} = DSV^m_n$ for any $1 < n < N$. For any client $c$, $DSV_c = DSV^m_{real}$. Also, for any $\langle k, h\rangle \in DS_c$, $h = DSV^m_{real}[k]$. Now we consider different operations and note how $DSV^m_n$, $DSV_c$, or $DS_c$ values changes. For sake of brevity, we assume $1 < n < N$, $1 < m < M \wedge i \neq m$, and $i \neq k$ implicit:

\begin{itemize}
\item \textbf{$DSV$ calculation}: At Line \ref{line:computeDSV} of Algorithm \ref{alg:server2}, each server computes the the $DSV$ as the entry-wise minimum of all $VVs$. Thus, by definition of $DSV^i_{real}$,  $DSV^i_{real} \geq DSV^m_n$.

\item \textbf{$GET$ operation}: By induction hypothesis, we have $dsv \leq DSV^i_{real}$. Thus, when the server updates its $DSV^i_n$ at Line \ref{line:get:takeMaxDSV} of Algorithm \ref{alg:server1}, still we have $DSV^i_n \leq DSV^i_{real}$. The server returns the value to client, and client updates its $DSV_c$ at Line \ref{line:client_update_dsv} of Algorithm \ref{alg:client} using $DSV^i_n$ received from server. Thus, still we have $DSV_c \leq DSV^i_{real}$. The client also updates its $DS_c$ with $ds$ received from server at Line \ref{line:client_update_ds} of Algorithm \ref{alg:client}. If the version is local, all of its dependencies are the version is smaller than $DSV_{real}^i$ by induction hypothesis. Otherwise, since for any $\langle k, h\rangle \in ds$, $h \leq DSV^i_n[k]$, and $DSV^i_n \leq DSV^i_{real}$, for any $\langle k, h\rangle$ in the new $DS_c$, we still have $h \leq DSV^i_{real}$.

\item \textbf{$PUT$ operation}: By induction hypothesis,  we know  for any $\langle k, h
 \rangle \in ds$, $h \leq DSV^i_{real}[k]$. Thus, when the server updates is $DSV$ at Line \ref{line:put_max_dsv} of Algorithm \ref{alg:server1}, we still have $DSV^i_n \leq DSV^i_{real}$.
 
\item \textbf{$ROTX$ operation}: By induction hypothesis, we know the $dsv$ and $ds$ received from client is less than or equal to $DSV^i_{real}$. Thus, when the server updates its $DSV^i_n$ at Line \ref{line:rotx_taking_max} of Algorithm \ref{alg:rotx}, still we have $DSV^i_n \leq DSV^i_{real}$. The server returns the values to client, and client updates its $DSV_c$ at Line \ref{line:rotx_update_dsv} of Algorithm \ref{alg:rotx} using $DSV^i_n$ received from server. Thus, still we have $DSV_c \leq DSV^i_{real}$. The client also updates its $DS_c$ with $ds$ received from server at Line \ref{line:client_update_ds} of Algorithm \ref{alg:client}. Since for any $\langle k, h\rangle \in ds$, $h \leq sv[k]$, and $sv \leq DSV^i_n \leq DSV^i_{real}$, for any $\langle k, h\rangle$ in the new $DS_c$, we still have $h \leq DSV^i_{real}$.
\end{itemize}
\end{proof}

We assume FIFO channels between replicas. Thus, we have the following observation 

\begin{observation}
\label{obs:fifo_real}
If a version $v$ written in data center $m$, has not arrived data center $i$, $DSV^i_{real}[m] < v.ut$.
\end{observation}

According to Observation \ref{obs:fifo_real}, and Lemma \ref{lem:real_is_greater}, we have the following observation: 

\begin{observation}
\label{obs:fifo}
If a version $v$ written in data center $m$, has not arrived data center $i$, for all $1 \leq n \leq N$, $DSV_n^i[m] < v.ut$. 
\end{observation}

According to Line \ref{line:rotx_setting_sv} of Algorithm \ref{alg:rotx}, we set the value of $sv$ for each $ROTX$ operation by a DSV value. Thus, according to Lemma \ref{lem:real_is_greater}, we have the following observation

\begin{observation}
\label{obs:sv_less_than_dsv_max}
For any $ROTX$ operation at data center $i$, for all $1 \leq m \leq M \wedge i \neq m$, $sv[m] \leq DSV^i_{real}[m]$.
\end{observation}

Note that in all lines where we change the DSV values, we always increase them. Thus, we note the following observation: 

\begin{observation}
\label{obs:dsv_does_not_decrease}
For all $1 \leq i \leq M$, $1 \leq n \leq N$, and $1 \leq m \leq M$, $DSV^i_n[m]$ never decreases. 
\end{observation}

\begin{lemma}
\label{lem:remains_visible_get}
Once a version becomes visible+ for $\{GET\}$ for a client in data center $i$ it remains visible+ for $\{GET\}$ for any client in data center $i$.
\end{lemma}
\begin{proof}
Suppose a version $v$ was visible+ to a client, but later it is not visible+. That requires that in response to a $GET$ operation, another version $v'$ is returned such that $f(v, v') = v$. According to Observations \ref{obs:f_property} and \ref{obs:dep_ut_ut}, $\langle v.ut, v.sr\rangle$ in any case is lexicographically higher than $\langle v'.ut, v'.sr\rangle$. Since $v$ was visible, it is either a local version, or for all $\langle i,h \rangle \in v.ds$, DSV had a greater value than $h$. According to Observation \ref{obs:dsv_does_not_decrease}, DSV never decreases, the condition is still valid. Thus, according to Line \ref{line:get:obtain} of Algorithm \ref{alg:server1}, it is impossible to return $v'$ (contradiction). 
\end{proof}


When a $GET$ operation returns a version, either the version is a local version, or the entries of DSV in the hosting partition is greater than the specified dependency for the version (see Line \ref{line:get:obtain} of Algorithm \ref{alg:server1}). Specifically, we have the following observation,

\begin{observation}
\label{obs:GET}
Let $v$ be a version for key $k$ that is written in partition $p^m_n$. If $GET(k)$ returns $v$ in data center $i \neq m$, then for any member $\langle j, h\rangle \in v.ds$, $DSV^i_n[j] \geq h$. 
\end{observation}

\begin{lemma}
\label{lem:local_dep_has_read}
Let $v_1$ be a version written in data center $i$. Then, for any non-local version $v_2$ such that $v_1 \deps v_2$, there must be a client $c$ at data center $i$ that has read either $v_2$, or a non-local version $v_3$ such that $v_3  \deps v_2$. 
\end{lemma}
\begin{proof}
According to Definition \ref{def:causal_dependecy}, (event of writing $v_2$) $\rightarrow$ (event of writing $v_1$). Since $v_1$ is a local version, and $v_2$ is not a local version, a non-local version $v_3$ must be read at data center $i$ in event $e$ such that (event of writing $v_2$) $\rightarrow e$ and $e \rightarrow$ (event of writing $v_1$). By definition \ref{def:happens}, (event of writing $v_3$) $\rightarrow e$. Thus, $v_3 = v_2$, or  (event of writing $v_2$) $\rightarrow$ (event of writing $v_3$) that leads to $v_3 \deps v_2$. 
\end{proof}
\begin{lemma}
\label{lem:local_dep_visibility}
Once client $c$ reads local version $v_1$ by a $GET$ operation, any version $v_2$ such that $v_1 \deps v_2$ is visible+ for $GET$ to $c$. 
\end{lemma}

\begin{proof}
Wlog assume $v_1$ and $v_2$ are versions for keys $k_1$, and $k_2$ written in partitions $n_1$ and $n_2$, respectively. Also, assume $v_1$ is written in data center $i$, and $v_2$ is written in data center $m$. 
According to Lemma \ref{lem:not_visible_get}, we have two cases: 

\begin{itemize}
\item [\textbf{Case 1}] $v_2$ has not arrived partition $p^i_{n_2}$ 

Since $v_1$ is a local version, and $v_2$ is not a local version, according to Lemma \ref{lem:local_dep_has_read}, we have two cases: 
\begin{itemize}
\item there is a client at data center $i$ that has read $v_2$

Contradiction to assumption that $v_2$ has not arrived partition $p^i_{n_2}$. 

\item there is a client at data center $i$ that has read non-local version $v_3$ such that $v_3 \deps v_2$

Wlog, assume $v_3$ is written in partition $n_3$. 

Since $v_2$ has not arrived data center $i$, according to Observation \ref{obs:fifo_real}, $DSV^i_{real}[m] < v_2.ut$. Since $v_3 \deps v_2$, according to Observation \ref{obs:dep_h_ut}, there is $\langle m , h\rangle \in v_3.ds$ such that $h \geq v_2.ut$. We have two cases for reading $v_3$: 
\begin{itemize}
\item $v_3$ has been read by a $GET$ operation: Since $v_3$ is not a local version, according to Observation \ref{obs:GET} for any $\langle m, h\rangle \in v_3.ds$, $DSV^m_{n_1}[m] \geq h$. According to Lemma \ref{lem:real_is_greater}, $DSV^i_{n_1}[m] \leq DSV^i_{real}[m]$. Thus, $DSV^i_{real}[m] \geq v_2.ut$ (contradiction).

\item $v_3$ has been read by an $ROTX$ operation:  for any $\langle m, h\rangle \in v_3.ds$, $sv[m] \geq h$. According to Observation \ref{obs:sv_less_than_dsv_max}, $sv[m] \leq DSV^i_{real}[m]$. Thus, $DSV^i_{real}[m] \geq v_2.ut$ (contradiction).
\end{itemize} 

\end{itemize}

\item [\textbf{Case 2}] there is a member $\langle k, h\rangle \in v_2.ds$, such that $DSV^i_{n_2}[k] < h$. 

 Since $v_1 \deps v_2$, according to Observation \ref{obs:dep_h'_h}, there is member $\langle k, h'\rangle \in v_1.ds$ such that $h' \geq h$. According to Line \ref{line:put_max_dsv} of Algorithm \ref{alg:server1}, we update the $DSV$ value with the set of dependecies when we write a new version in a $PUT$ operation. Thus, $DSV^i_{n_1}[k] \geq h'$ that leads to $DSV^i_{n_1}[k] \geq h$. Since client read $v_1$ before reading $k_2$, the $DSV^i_{n_2}[k]$ at time of reading $k_2$ is greater or equal to $DSV^i_{n_1}$, according to Lines \ref{line:client_update_dsv} of Algorithm \ref{alg:client}, and Line \ref{line:get:takeMaxDSV} of Algorithm \ref{alg:server1}. Thus, $DSV^i_{n_2}[k] \geq h$ (contradiction). 
\end{itemize}
\end{proof}

\begin{theorem}
\label{lem:p_for_get_get}
The data store running \name protocol defined in Section \ref{sec:protocol} is causal+ consistent for $\{GET\}$ with reader operations $\{GET\}$ for $f =$ last-writer-wins. 
\end{theorem}

\begin{proof}
We prove this theorem by showing that \name protocol satisfies Definition \ref{def:causal+}. 

The first condition is satisfied based on the fact that clients are sticky and Lemma \ref{lemma:local_visibility}. The second condition  is satisfied based on Lemma \ref{lem:remains_visible_get}. 
We prove the third condition via contradiction: 

Let $k_1$ and $k_2$ be any two arbitrary keys in the store residing in partitions $n_1$ and $n_2$. Let $v_1$ be a version of key $k_1$, and $v_2$ be a version of key $k_2$ such that $v_1 \deps v_2$. Now suppose client $c$ reads $v_1$ in data center $i$ via a $GET$ operation, but $v_2$ is not visible+ for a $GET$ operation to client $c$. 
%
%
Let $v_2$ be a version written in a data center $m$. According to Lemma \ref{lem:not_visible_get}, two cases are possible: 
\begin{itemize}
\item  [\textbf{Case 1}] $v_2$ has not arrived partition $p^i_{n_2}$.

Since $v_2$ has not arrived at data center $i$, according to  Observation \ref{obs:fifo}, $DSV^i_{n_1}[m]< v_2.ut$. Since, $v_1 \deps v_2$, according to Observation \ref{obs:dep_h_ut}, $v_1.ds$ has member $\langle m, h\rangle$ such that $h \geq v_2.ut$. According to Lemma \ref{lem:local_dep_visibility}, $v_1$ is not a local version. Since $GET(k_1) = v_1$, and $v_1$ is not a local version, according to Observation \ref{obs:GET}, for any member $\langle j, h\rangle \in v.ds$, $DSV^i_{n_1}[j] \geq h$. Thus, $DSV^i_{n_1}[m] \geq v_2.ut$ (contradiction).

\item [\textbf{Case 2}] $v_2$ has arrived, but there is a member $\langle k, h\rangle \in v_2.ds$, such that $DSV^i_{n_2}[k] < h$. 

Since client $c$ asks for $k_2$ after reading $v_1$, according to Line \ref{line:get:takeMaxDSV} of Algorithm \ref{alg:server1}, $DSV^i_{n_2}[k]$ at the time of reading $k_2$ is higher than $DSV^i_{n_1}[k]$ at the time of reading $k_1$ ($DSV^i_{n_1}[k] \leq DSV^i_{n_2}[k]$). Since $v_1 \deps v_2$, according to Observation \ref{obs:dep_h'_h}, there is member $\langle k, h'\rangle \in v_1.ds$ such that $h' \geq h$. According to Lemma \ref{lem:local_dep_visibility}, $v_1$ is not a local version. Since $GET(k_1) = v_1$, and $v_1$ is not a local version, according to Observation \ref{obs:GET}, $DSV^i_{n_1}[k] \geq h'$ that leads to $DSV^i_{n_2}[k] \geq h$ (contradiction).

\end{itemize}


\end{proof}



With Lemma \ref{lemma:local_visibility} and Theorem \ref{lem:p_for_get_get}, we have the following corollary: 
\begin{corollary}
\label{th:pp_for_get_get}
The data store running \name protocol defined in Section \ref{sec:protocol} is causal++ consistent for $\{GET\}$ for reader operations $\{GET\}$ and instantly visible operations $\{GET\}$ for conflict resolution function $f = \textsf{last-writer-wins}$. 
\end{corollary}

\begin{lemma}
\label{lem:remains_visible_rotx_get}
Once a client $c$ reads version $v$ of key $k$ at partition $n$ by an $\{ROTX\}$  operation at data center $i$, it remains visible+ for $\{GET\}$ for $c$.
\end{lemma}
\begin{proof}
Since the client has read the version, it is obviously written. According to Lemma \ref{lemma:local_visibility}, $v_2$ is not a local version. Now, according to Lemma \ref{lem:not_visible_get}, two cases are possible: 
\begin{itemize}
\item  [\textbf{Case 1}] $v$ has not arrived partition data center $i$.

Contradiction to the assumption that the client has read the version. 

\item [\textbf{Case 2}] $v$ has arrived, but there is a member $\langle j, h\rangle \in v.ds$, such that $DSV^i_{n}[j] < h$.

Since client $c$ asks for $k$ after reading $v$ by an $ROTX$ operations, according to Line \ref{line:get:takeMaxDSV} of Algorithm \ref{alg:server1}, $DSV^i_{n}[j]$ at the time of reading $k$ using $GET$ is higher than or equal to $DSV_c[j]$ after reading $k$ by $ROTX$ operation. Since $ROTX$ has read $v_1$, according to Line \ref{line:obtain_sv} of Algorithm \ref{alg:rotx}, $sv[j] \geq h$. Since $sv[j] \leq DSV_{c}[j]$, $DSV_{c}[j] \geq h$ that leads to $DSV^i_{n}[j] \geq h$ at the time $GET$ operation (contradiction).

\end{itemize}

\end{proof}

\begin{theorem}
\label{lem:p_for_get_rotx}
The data store running \name protocol defined in Section \ref{sec:protocol} is causal+ consistent for $\{GET\}$ with reader operations $\{ROTX\}$ for $f =$ last-write-wins. 
\end{theorem}

\begin{proof}
We prove this theorem by showing that \name protocol satisfies Definition \ref{def:causal+}.

The first condition is satisfied based on the fact that clients are sticky and Lemma \ref{lemma:local_visibility}. The second condition  is satisfied based on Lemma \ref{lem:remains_visible_rotx_get}. 
We prove the third condition via contradiction: 

Let $k_1$ and $k_2$ be any two arbitrary keys in the store residing in partitions $n_1$ and $n_2$. Let $v_1$ be a version of key $k_1$, and $v_2$ be a version of key $k_2$ such that $v_1 \deps v_2$. Now suppose client $c$ reads $v_1$ in data center $i$ via an $ROTX$ operation, but $v_2$ is not visible+ for a $GET$ operation to client $c$. 

According to Lemma \ref{lemma:local_visibility}, $v_2$ is not a local version. Let $v_2$ be a version written in a data center $m$. According Lemma \ref{lem:not_visible_get}, two cases are possible: 
\begin{itemize}
\item  [\textbf{Case 1}] $v_2$ has not arrived partition $p^i_{n_2}$. 

Since $v_2$ has not arrived at data center $i$, according to Observation \ref{obs:fifo_real}, $DSV^i_{real}[m]< v_2.ut$. Since, $v_1 \deps v_2$, according to Observation \ref{obs:dep_h_ut}, $v_1.ds$ has member $\langle m, h\rangle$ such that $h \geq v_2.ut$. Since an $ROTX$ operation has read  $v_1$, for any member $\langle j, h\rangle \in v.ds$, $sv[j] \geq h$. According to Observation \ref{obs:sv_less_than_dsv_max}, $sv[m] \leq DSV^i_{real}[m]$. Thus, $DSV^i_{real}[m] \geq v_2.ut$ (contradiction).  

\item [\textbf{Case 2}] $v_2$ has arrived, but there is a member $\langle k, h\rangle \in v_2.ds$, such that $DSV^i_{n_2}[k] < h$.

Since client $c$ asks for $k_2$ after reading $v_1$, according to Line \ref{line:get:takeMaxDSV} of Algorithm \ref{alg:server1}, $DSV^i_{n_2}[k]$ at the time of reading $k_2$ is higher than $DSV_{c}[k]$ after reading $k_1$. Since $v_1 \deps v_2$, according to Observation \ref{obs:dep_h'_h}, there is member $\langle k, h'\rangle \in v_1.ds$ such that $h' \geq h$. Since $ROTX$ has read $v_1$, according to Line \ref{line:obtain_sv} of Algorithm \ref{alg:rotx}, $sv[k] \geq h'$. Since $sv \leq DSV_{c}$, $DSV_{c}[k] \geq h'$ that leads to $DSV^i_{n_2}[k] \geq h$ (contradiction).
\end{itemize}
\end{proof}

With Lemma \ref{lemma:local_visibility} and Theorem \ref{lem:p_for_get_rotx}, we have the following corollary: 

\begin{corollary}
\label{th:pp_for_get_rotx}
The data store running \name protocol defined in Section \ref{sec:protocol} is causal++ consistent for $\{GET\}$ for reader operations $\{ROTX\}$ and instantly visible operations $\{GET\}$ for conflict resolution function $f = \textsf{last-writer-wins}$. 
\end{corollary}

According to Corollaries \ref{th:pp_for_get_get} and \ref{th:pp_for_get_rotx}, and Observation \ref{obs:comp_R}, we have the following corollary: 

\begin{corollary}
\label{cor:pp_for_get}
The data store running \name protocol defined in Section \ref{sec:protocol} is causal++ consistent for $\{GET\}$ for reader operations $\{GET, ROTX\}$ and instantly visible operations $\{GET\}$ for conflict resolution function $f =$ last-write-wins. 
\end{corollary}


\subsubsection{Causal Consistency for ROTX operations}

\begin{lemma}
\label{lem:not_visible_rotx}
Let $v$ be a version that is written in partition $p^m_n$. If $v$ is not visible+ for $\{ROTX\}$ for conflict resolution function $f=$ last-write-wins to a client in data center $i$, then
\begin{itemize}
\item $v$ has not arrived partition $p^i_n$, or 
\item there is a member $\langle j, h\rangle \in v.ds$, such that $sv[j] < h$, or
\item $sv[v.sr] < v.ut.$ 
\end{itemize}

\end{lemma}

\begin{proof}
The poof of this lemma is the same as that of Lemma \ref{lem:not_visible_get}, except we must note Line \ref{line:obtain_sv} of Algorithm \ref{alg:rotx}.
\end{proof}

\begin{lemma}
\label{lem:self-writtens_are_visible_to_rotx}
All versions written by client $c$ are visible+ for $\{ROTX\}$ to client $c$. 
\end{lemma}
\begin{proof}
Suppose version $v$ written by client $c$ is not visible to $ROTX$. Suppose client $c$ accesses data center $i$. Now, according to Lemma \ref{lem:not_visible_rotx}, we have following cases: 

\begin{itemize}
\item [\textbf{Case 1}] $v$ has not arrived data center $i$
Impossible, since the version is written by client $c$

\item there is a member $\langle j, h\rangle \in v.ds$, such that $sv[j] < h$
Since the client has written version $v$, for any $\langle j, h\rangle \in v.ds$, $DS_c[j] > h$. According to Line \ref{line:rotx_taking_max}, $sv[j] > h$ (contradiction).

\item $sv[v.sr] < v.ut.$ 
Since the client has written version $v$, $DS_c[v.sr] > v.ut$. According to Line \ref{line:rotx_taking_max}, $sv[j] > v.ut$ (contradiction). 

\end{itemize}

\end{proof}

\begin{lemma}
\label{lem:local_dep_visibility_for_rotx}
Once client $c$ reads \textit{local} version $v_1$ by a $GET$ operation, any version $v_2$ such that $v_1 \deps v_2$ is visible+ for $\{ROTX\}$ to $c$. 
\end{lemma}

\begin{proof}
Wlog assume $v_1$ and $v_2$ are versions for keys $k_1$, and $k_2$ written in partitions $n_1$ and $n_2$, respectively. Also, assume $v_2$ is written in data center $m$. If version $v_2$ is not visible+ for $\{ROTX\}$, Lemma \ref{lem:not_visible_rotx}, we have three cases:
\begin{itemize}
\item $v_2$ has not arrived partition $p^i_n$

In this case, with the same argument provided in Case 1 of proof of Lemma \ref{lem:local_dep_visibility}, we face a contradiction. 

\item there is a member $\langle k, h\rangle \in v_2.ds$, such that $sv[k] < h$. 

Since $v_1 \deps v_2$, according to Observation \ref{obs:dep_h'_h}, there is member $\langle k, h'\rangle \in v_1.ds$ such that $h' \geq h$. According to Line \ref{line:put_max_dsv} of Algorithm \ref{alg:server1}, we update the $DSV$ value with the set of dependencies when we write a new version in a $PUT$ operation. Thus, $DSV^i_{n_1}[k] \geq h'$ that leads to $DSV^i_{n_1}[k] \geq h$. Since client read $v_1$ before reading $k_2$, the $DSV_c[k]$ at time of reading $k_2$ is greater or equal to $DSV^i_{n_1}$. 
According to Line \ref{line:rotx_taking_max}, and \ref{line:rotx_setting_sv} of Algorithm \ref{alg:rotx}, the $sv$ of the $ROTX$ operation is entry-wise greater than $DSV_c$. Thus, $sv[k] \geq h$ (contradiction).

\item $sv[v_2.sr] < v_2.ut$.\\
Since $v_1 \deps v_2$, according to Observation \ref{obs:dep_h_ut}, $\langle v_2.sr, h\rangle \in v_1.ds$ such that $h \geq v_2.ut$. Since client read $v_1$ before reading $k_2$, $ds[v.sr] \geq h$ for $ROTX$ operation reading $k_2$. According to Line \ref{line:rotx_taking_max} and \ref{line:rotx_setting_sv}, $sv[v_2.sr] \geq v_2.ut$ (contradiction). 

\end{itemize}
\end{proof}

\begin{lemma}
\label{lem:remains_visible_get_rotx}
When client $c$ reads a version $v$ using $GET$ operation, it remains visible+ for $\{ROTX\}$ for $c$.
\end{lemma}
\begin{proof}
Suppose $v$ is not visible+ for $\{ROTX\}$ to client $c$. Since client has read the version before, obviously it is written. According to Lemma \ref{lem:not_visible_rotx}, there are three cases: 

\begin{itemize}
\item [\textbf{Case 1}] $v$ has not arrived the hosting partition. 

Contradiction to the assumption that the client has read the version. 

\item [\textbf{Case 2}]  there is a member $\langle j, h\rangle \in v.ds$, such that $sv[j] < h$ 

Since the client has read $v$, the $DSV_c[j] \geq h$. According to Line \ref{line:rotx_taking_max}, and \ref{line:rotx_setting_sv} of Algorithm \ref{alg:rotx}, the $sv$ of the $ROTX$ operation is entry-wise greater than $DSV_c$. Thus, $sv[j] \geq h$ (contradiction). 

\item [\textbf{Case 3}]  $sv[v.sr] < v.ut$.\\
Since the client has read $v$, the $DS_c[v.sr] \geq v.ut$ after reading $v$. According to Line \ref{line:rotx_taking_max}, and \ref{line:rotx_setting_sv} of Algorithm \ref{alg:rotx}, the $sv$ of the $ROTX$ operation is entry-wise greater than $DS_c$. Thus, $sv[v.sr] \geq v.ut$ (contradiction). 
\end{itemize}
\end{proof}

\begin{theorem}
\label{lem:p_for_rotx_get}
The data store running \name protocol defined in Section \ref{sec:protocol} is causal+ consistent for $\{ROTX\}$ with reader operations $\{GET\}$ for $f =$ last-write-wins. 
\end{theorem}
\begin{proof}
We prove this theorem by showing that \name protocol satisfies Definition \ref{def:causal+}.

The first condition is satisfied according to Lemma \ref{lem:self-writtens_are_visible_to_rotx}. The second condition  is satisfied based on Lemma \ref{lem:remains_visible_get_rotx}. 
We prove the third condition via contradiction: 

Let $k_1$ and $k_2$ be any two arbitrary keys in the store residing in partitions $n_1$ and $n_2$. Let $v_1$ be a version of key $k_1$, and $v_2$ be a version of key $k_2$ such that $v_1 \deps v_2$. Now suppose client $c$ reads $v_1$ in data center $i$ via a $GET$ operation, but $v_2$ is not visible+ for  $\{ROTX\}$ operation to client $c$. 

Let $v_2$ be a version written in a data center $m$. According to  Lemma \ref{lem:not_visible_rotx}, three cases are possible: 
\begin{itemize}
\item  [\textbf{Case 1}] $v_2$ has not arrived partition $p^i_{n_2}$. 

Since $v_2$ has not arrived at data center $i$, according to Observation \ref{obs:fifo_real},  $DSV^i_{real}[m]< v_2.ut$, and according to Observation \ref{obs:fifo}, $DSV^i_{n_1}[m]< v_2.ut$. Since, $v_1 \deps v_2$, according to Observation \ref{obs:dep_h_ut}, $v_1.ds$ has member $\langle m, h\rangle$ such that $h \geq v_2.ut$. According to Lemma \ref{lem:local_dep_visibility_for_rotx}, $v_1$ is not a local version. Since $GET(k_1) = v_1$, and $v_1$ is not a local version, according to Observation \ref{obs:GET}, for any member $\langle j, h\rangle \in v.ds$, $DSV^i_{n_1}[j] \geq h$. Thus, $DSV^i_{n_1}[m] \geq v_2.ut$ (contradiction).

\item [\textbf{Case 2}] there is a member $\langle j, h\rangle \in v.ds$, such that $sv[j] < h$. 

Since $v_1 \deps v_2$, according to Observation \ref{obs:dep_h'_h}, for any member $\langle j, h\rangle \in v_2.ds$, there exists $\langle j, h'\rangle \in v_1.ds$ such that $h' \geq h$. Since the client has read $v_1$ by a $GET$ operations, the $DSV_c[j] \geq h'$. According to Line \ref{line:rotx_taking_max}, and \ref{line:rotx_setting_sv} of Algorithm \ref{alg:rotx}, the $sv$ of the $ROTX$ operation is entry-wise greater than $DSV_c$. Thus, $sv[j] \geq h'$. This leads to $sv[j] \geq h$ (contradiction). 

\item [\textbf{Case 3}] $sv[v_2.sr] < v_2.ut$.\\
Since $v_1 \deps v_2$, according to Observation \ref{obs:dep_h_ut}, $\langle v_2.sr, h\rangle \in v_1.ds$ such that $h \geq v_2.ut$. Since client read $v_1$ before reading $k_2$, $ds[v.sr] \geq h$ for $ROTX$ operation reading $k_2$. According to Line \ref{line:rotx_taking_max} and \ref{line:rotx_setting_sv} of Algorithm \ref{alg:rotx}, $sv[v_2.sr] \geq v_2.ut$ (contradiction). 

\end{itemize}

\end{proof}

With Lemma \ref{lemma:local_visibility} and Theorem \ref{lem:p_for_rotx_get}, we have the following corollary: 

\begin{corollary}
\label{th:pp_for_rotx_get}
The data store running \name protocol defined in Section \ref{sec:protocol} is causal++ consistent for $\{ROTX\}$ for reader operations $\{GET\}$ and instantly visible operations $\{GET\}$ for conflict resolution function $f = \textsf{last-writer-wins}$. 
\end{corollary}

\begin{lemma}
\label{lem:remains_visible_rotx_rotx}
When client $c$ read a version $v$ using $ROTX$ operation, it remains visible+ for $\{ROTX\}$ for $c$.
\end{lemma}
\begin{proof}
Suppose $v$ is not visible+ for $\{ROTX\}$ to client $c$. Since client has read the version before, and  according to Lemma \ref{lem:not_visible_rotx}, there are two cases: 

\begin{itemize}
\item [\textbf{Case 1}]  $v$ has not arrived the hosting partition. 

Contradiction to the assumption that the client has read the version. 

\item [\textbf{Case 2}]   there is a member $\langle j, h\rangle \in v.ds$, such that $sv[j] < h$ 

Since the client has read $v$, the $DSV_c[j] \geq h$. According to Line \ref{line:rotx_taking_max}, and \ref{line:rotx_setting_sv} of Algorithm \ref{alg:rotx}, the $sv$ of the $ROTX$ operation is entry-wise greater than $DSV_c$. Thus, $sv[j] \geq h$ (contradiction). 

\item [\textbf{Case 3}]  $sv[v.sr] < v.ut$.\\
Since the client has read $v$, the $DS_c[v.sr] \geq v.ut$ after reading $v$. According to Line \ref{line:rotx_taking_max}, and \ref{line:rotx_setting_sv} of Algorithm \ref{alg:rotx}, the $sv$ of the $ROTX$ operation is entry-wise greater than $DS_c$. Thus, $sv[v.sr] \geq v.ut$ (contradiction). 

\end{itemize}
\end{proof}

\begin{theorem}
\label{lem:p_for_rotx_rotx}
The data store running \name protocol defined in Section \ref{sec:protocol} is causal+ consistent for $\{ROTX\}$ with reader operations $\{ROTX\}$ for $f = \textsf{last-writer-wins}$. 
\end{theorem}

\begin{proof}
We prove this theorem by showing that \name protocol satisfies Definition \ref{def:causal+}.

The first condition is satisfied according to Lemma \ref{lem:self-writtens_are_visible_to_rotx}. The second condition  is satisfied based on Lemma \ref{lem:remains_visible_rotx_rotx}. 
We prove the third condition via contradiction: 

Let $k_1$ and $k_2$ be any two arbitrary keys in the store residing in partitions $n_1$ and $n_2$. Let $v_1$ be a version of key $k_1$, and $v_2$ be a version of key $k_2$ such that $v_1 \deps v_2$. Now suppose client $c$ reads $v_1$ in data center $i$ via an $ROTX$ operation, but $v_2$ is not visible+ for $\{ROTX\}$ operation to client $c$. 

Let $v_2$ be a version written in a data center $m$. According to Lemma \ref{lem:not_visible_rotx}, three cases are possible: 
\begin{itemize}
\item [\textbf{Case 1}] $v_2$ has not arrived partition $p^i_{n_2}$. 

Since $v_2$ has not arrived at data center $i$, according to Observation \ref{obs:fifo_real}, $DSV^i_{real}[m]< v_2.ut$. Since, $v_1 \deps v_2$, according to Observation \ref{obs:dep_h_ut}, $v_1.ds$ has member $\langle m, h\rangle$ such that $h \geq v_2.ut$. Since an $ROTX$ operation has read  $v_1$, for any member $\langle j, h\rangle \in v.ds$, $sv[j] \geq h$. According to Observation \ref{obs:sv_less_than_dsv_max}, $sv < DSV^i_{real}$. Thus, $DSV_{real}[m] \geq v_2.ut$ (contradiction).

\item [\textbf{Case 2}] there is a member $\langle j, h\rangle \in v.ds$, such that $sv[j] < h$. 

Since $v_1 \deps v_2$, according to Observation \ref{obs:dep_h'_h}, for any member $\langle j, h\rangle \in v_2.ds$, there exists $\langle j, h'\rangle \in v_1.ds$ such that $h' \geq h$. Since the client has read $v_1$ by a $ROTX$ operations, the $sv[j] \geq h'$ where $sv_1$ is that $sv$ of the $ROTX$ reading $v_1$. Since $sv_1 < DSV_c$, $DSV_c[j] > h'$. According to Line \ref{line:rotx_taking_max}, and \ref{line:rotx_setting_sv} of Algorithm \ref{alg:rotx}, the $sv$ of the $ROTX$ operation is entry-wise greater than $DSV_c$. Thus, $sv[j] \geq h'$. This leads to $sv[j] \geq h$ (contradiction). 

\item [\textbf{Case 3}] $sv[v_2.sr] < v_2.ut$.\\
Since $v_1 \deps v_2$, according to Observation \ref{obs:dep_h_ut}, $\langle v_2.sr, h\rangle \in v_1.ds$ such that $h \geq v_2.ut$. Since client read $v_1$ before reading $k_2$, $ds[v.sr] \geq h$ for $ROTX$ operation reading $k_2$. According to Line \ref{line:rotx_taking_max} and \ref{line:rotx_setting_sv} of Algorithm \ref{alg:rotx}, $sv[v_2.sr] \geq v_2.ut$ (contradiction). 

\end{itemize}
\end{proof}

With Lemma \ref{lemma:local_visibility} and Theorem \ref{lem:p_for_rotx_rotx}, we have the following corollary: 

\begin{corollary}
\label{th:pp_for_rotx_rotx}
The data store running \name protocol defined in Section \ref{sec:protocol} is causal++ consistent for $\{ROTX\}$ for reader operations $\{ROTX\}$ and instantly visible operations $\{GET\}$ for conflict resolution function $f = \textsf{last-writer-wins}$. 
\end{corollary}

According to Corollaries \ref{th:pp_for_rotx_get} and \ref{th:pp_for_rotx_rotx} and Observation \ref{obs:comp_R} in Section \ref{sec:causalConsistency}, we have the following corollary: 

\begin{corollary}
\label{cor:pp_for_rotx}
The data store running \name protocol defined in Section \ref{sec:protocol} is causal++ consistent for $\{ROTX\}$ for reader operations $\{GET, ROTX\}$ and instantly visible operations $\{GET\}$ for conflict resolution function $f = \textsf{last-writer-wins}$. 
\end{corollary}

Now, based on Corollaries \ref{cor:pp_for_get} and \ref{cor:pp_for_rotx} and Observation \ref{obs:comp_O}, we conclude our desired consistency requirement. Specifically, 
\begin{corollary}
\label{cor:pp_for_get_and_rotx}
The data store running \name protocol defined in Section \ref{sec:protocol} is causal++ consistent for $\{GET, ROTX\}$ for reader operations $\{GET, ROTX\}$ and instantly visible operations $\{GET\}$ for conflict resolution function $f = \textsf{last-writer-wins}$. 
\end{corollary}

\subsubsection{Causal Consistency of Values Returned by $ROTX$ Operations}

We formally define this requirement as follows. First, we define visiblity+ for a set of versions: 

\begin{definition} [Visible+ to a Set of Versions]
Let $vset$ be a set of version. We say version $v$ is visible+ for conflict resolution function $f$ to $vset$, if $vset$ does not include version $v'$ such that $v' \neq  v$ and $v = f(v, v')$. 
\end{definition}

Now, we define causal consistency for a set of versions: 

\begin{definition} [Causal Consistency for a Set of Versions]
Let $vset$ be a set of versions. $vest$ is causally consistent for conflict resolution function $f$, if for any version $v_1 \in vset$, any version $v_2$ such that $v_1 \deps v_2 $ is visible+ for $f$ to $vest$. 
\end{definition}

A group of values returned by any $ROTX$ is causally consistent. Specifically,

\begin{theorem}
\label{the:consistency_vset_rotx}
Let $vset$ be a set of versions returned by an $ROTX$ operation. $vset$ is causally consistent for $f =\textsf{last-writer-wins}$.
\end{theorem}
\begin{proof}
We prove this theorem by contradiction: 

Let $kset$ and $vset$ respectively be the set of keys read and the set of of versions returned by $ROTX$ operation $a$ performed by client $c$ at data center $i$. Let $v_1$ and $v_2$ be  versions for key $k_1 \in kset$ and $k_2\in kset$, respectively, such that $v_1 \deps v_2$. Suppose $v_1 \in vset$, but $v_2$ is not visible+.
When $v_2$ is not visible+ to $vset$, it means it is no visible+ for $\{ROTX\}$ to client $c$. 

We have following cases for $v_2$:

\begin{itemize}
\item [\textbf{Case 1}] $v_2$ is not written  when $ROTX$ reads $k_2$: \\ 
If $v_2$ is a remote version, this case is the same as \textbf{Case 2.1} where $v_2$ has not been arrived. Thus, we focus on the situation where $v_2$ is a local version. Now, we consider two cases for $v_1$: 
\begin{itemize}
\item [\textbf{Case 1.1}] $v_1$ is a remote version:\\
Suppose $v_1$ is written in data center $j$. Since $v_1 \deps v_2$ and $v_2$ is not written at the time of $ROTX$ begins, $v_1$ also is not written at that time. Thus, data center could not have received $v_1$ when it starts the $ROTX$ operation. According to Observation \ref{obs:fifo}, $DSV^i_{real}[j] < v_1.ut$. According to Observation \ref{obs:sv_less_than_dsv_max}, $sv[j] < v_1.ut$. However, since $ROTX$ read $v_1$, $sv[j] \geq v_1.ut$ (contradiction).

\item [\textbf{Case 1.2}] $v_1$ is a local version:\\
According to Line \ref{line:ubdateDsvWithSv} of Algorithm \ref{alg:rotx}, $DSV^i_{n_2}[i] \geq sv[i]$, after $ROTX$ read $k_2$. Since $v_2$ is written after $ROTX$ read $k_2$, according to Line \ref{line:takeMaxforDt} of Algorithm \ref{alg:server1} and the algorithm of $updateHLC$, $v_2.ut > DSV^i_{n_2}[i]$ that leads to $v_2.ut > sv[i]$. Since $v_1 \deps v_2$, according to Observation \ref{obs:dep_h_ut}, $\langle i, h \rangle \in v_1.ds$  such that $h \geq v_2.ut$. Thus, $h > sv[i]$. However, since $ROTX$ reads $v_1$, $sv[i] \geq h$ (contradiction).
\end{itemize}
 
\item [\textbf{Case 2}] $v_2$ is written when $ROTX$ reads $k_2$: \\
In this case, according to Lemma \ref{lem:not_visible_rotx}, two cases are possible for $v_2$:

\begin{itemize}

\item [\textbf{Case 2.1}] $v_2$ has not arrived partition $p^i_{n_2}$. 

Since $v_2$ has not arrived at data center $i$, according to Observation \ref{obs:fifo_real}, $DSV_{real}[m]< v_2.ut$. Since, $v_1 \deps v_2$, according to Observation \ref{obs:dep_h_ut}, $v_1.ds$ has member $\langle m, h\rangle$ such that $h \geq v_2.ut$. Since the $ROTX$ operation has read  $v_1$, for any member $\langle j, h\rangle \in v.ds$, $sv[j] \geq h$. According to Observation \ref{obs:sv_less_than_dsv_max}, $sv \leq DSV^i_{real}$. Thus, $DSV^i_{real}[m] \geq v_2.ut$ (contradiction).

\item [\textbf{Case 2.2}] there is a member $\langle j, h\rangle \in v_2.ds$, such that $sv[j] < h$. 

Since $v_1 \deps v_2$, according to Observation \ref{obs:dep_h'_h}, for any member $\langle j, h\rangle \in v_2.ds$, there exists $\langle j, h'\rangle \in v_1.ds$ such that $h' \geq h$. Since the client has read $v_1$ by the $ROTX$ operation,  $sv[j] \geq h'$. This leads to $sv[j] \geq h$ (contradiction). 

\item [\textbf{Case 2.3}] $sv[v_2.sr] < v_2.ut$.\\
Since $v_1 \deps v_2$, according to Observation \ref{obs:dep_h_ut}, $\langle v_2.sr, h\rangle \in v_1.ds$ such that $h \geq v_2.ut$. Since $ROTX$ reads $v_1$, $sv[v_2.sr] \geq h$ which leads to $sv[v_2.sr] \geq v_2.ut$ (contradiction). 
\end{itemize}
\end{itemize}
\end{proof}

\subsection{Convergence}
\subsubsection{Convergence for $GET$ Operations}
\label{sec:proof_conv_get}
Now, we focus on the convergence aspect of the protocol. First, we observe that DSV values for connected servers never stop increasing via replicate or heartbeat messages. Thus, we have

\begin{observation}
\label{obs:dsvGrows}
If data center $i$ is connected to data center $j$, for all $1 \leq k \leq N$, $DSV^j_k[i]$ and $DSV^i_k[j]$ will never stop increasing.
\end{observation}

We have the following theorem about the convergence of our protocol, 

\begin{theorem}
\label{the:convergence_get}
The data store running \name protocol defined in Section \ref{sec:protocol} is convergent for $\{GET\}$ for conflict resolution function \textsf{last-writer-wins}.
\end{theorem}

\begin{proof}
We prove this theorem by showing that \name protocol satisfies Definition \ref{def:convergence} via proof by contradiction: 

Let $v_1$ be a version for key $k$ written in partition $p^m_n$. Let 
\begin{itemize}
\item data center $i$ be connected to data center $m$, and 
\item for any version $v_2$ such that $v_1 \deps v_2$, data center $i$ be connected to data center $j$ where version $v_2$ is written. 
\end{itemize}
Now, suppose $v_1$ is never visible+ in $i$. 

When $v_1$ is never visible+ in $i$, according to Lemma \ref{lem:not_visible_get}, two cases are possible for $v_1$: 

\begin{itemize}
\item [\textbf{Case 1}]$v_1$ will never reach partition $p^i_n$. 

Since data center $i$ is connected to data center $m$, $p^i_n$ will receive replicate message sent in Line \ref{line:sendReplicate} of Algorithm \ref{alg:server1} (contradiction).
\item [\textbf{Case 2}]there is a member $\langle j, h\rangle \in v_1.ds$, such that $DSV^i_n[j]$ will always remain smaller than $h$. 

This is a contradiction to Observation \ref{obs:dsvGrows} and the second assumption of the theorem. 
\end{itemize}
\end{proof}

\subsubsection{Convergence for ROTX Operations}
\begin{theorem}
\label{the:convergence_rotx}
The data store running \name protocol defined in Section \ref{sec:protocol} is convergent for $\{ROTX\}$ for conflict resolution function \textsf{last-writer-wins}.
\end{theorem}

\begin{proof}
We prove this theorem by showing that \name-X protocol satisfies Definition \ref{def:convergence} via proof by contradiction: 

Let $v_1$ be a version for key $k$ written in partition $p^m_n$. Let 
\begin{itemize}
\item data center $i$ be connected to data center $m$, and 
\item for any version $v_2$ that $v_1 \deps v_2$, data center $i$ be connected to data center $j$ where version $v_2$ is written. 
\end{itemize}
Now, suppose $v_1$ is never visible+ in $i$. 

When $v_1$ is never visible+ in $i$, according to Lemma \ref{lem:not_visible_rotx}, following cases are possible for $v_1$: 

\begin{itemize}
\item [\textbf{Case 1}]$v_1$ will never reach at partition $p^i_n$. \\
Since data center $i$ is connected to data center $m$, $p^i_n$ will receive replicate message sent in Line \ref{line:sendReplicate} of Algorithm \ref{alg:server1} (contradiction).

\item [\textbf{Case 2}]there is a member $\langle j, h\rangle \in v_1.ds$, such that $sv[j]$ will always remain smaller than $h$, or $sv[v_1.sr]$ always remains smaller than $v_1.ut$\\
This is a contradiction to Observation \ref{obs:dsvGrows} and the second assumption of the theorem, and the fact that $sv$ is entry-wise greater than DSV. 

\end{itemize}
\end{proof}

According to Theorem \ref{the:convergence_get} provided in Section \ref{sec:proof_conv_get}, Theorem \ref{the:convergence_rotx} provided in this section, and Observation \ref{obs:comp_conv_O} provided in Section \ref{sec:background}, we conclude our desired convergence requirement. Specifically,

\begin{corollary}
\label{cor:conv_for_get_rotx}
The data store running \name protocol defined in Section \ref{sec:protocol} is convergent for $\{GET, ROTX\}$ for conflict resolution function \textsf{last-writer-wins}.
\end{corollary}

Corollaries \ref{cor:pp_for_get_and_rotx} and \ref{cor:conv_for_get_rotx} and Theorem \ref{the:consistency_vset_rotx} provide the correctness of \name.